\RequirePackage{fix-cm} 
\documentclass[conference,compsoc]{IEEEtran}
\ifCLASSOPTIONcompsoc
  \usepackage[nocompress]{cite}
\else
  \usepackage{cite}
\fi
\ifCLASSINFOpdf
\else
\fi
\usepackage{array}
\usepackage{amsmath}
\usepackage{amssymb}
\usepackage[ruled,vlined,linesnumbered]{algorithm2e}
\usepackage{enumitem}
\usepackage{booktabs}
\usepackage[table]{xcolor}
\usepackage{tcolorbox}
\usepackage[T1]{fontenc}
\usepackage{inconsolata}
\usepackage{etoc}
\usepackage{xurl}
\usepackage[hidelinks]{hyperref}

\usepackage{amsthm}
\theoremstyle{plain}
\newtheorem{proposition}{Proposition}
\newtheorem{lemma}{Lemma}
\newtheorem{theorem}{Theorem}

\def\MethodName{\textsc{MMGuard}}

\tcbuselibrary{skins}
\tcbset{
    promptstyle/.style={
        colback=gray!5,
        colframe=black,
        fonttitle=\bfseries,
        boxrule=0.5mm,
        sharp corners,
        enhanced,
        width=\linewidth
    }
}

\begin{document}
%
\title{To See is Not to Learn: Protecting Multimodal Data from Unauthorized Fine-Tuning of Large Vision-Language Models}

\author{\IEEEauthorblockN{Chengshuai Zhao\IEEEauthorrefmark{1},
Zhen Tan\IEEEauthorrefmark{2},
Dawei Li\IEEEauthorrefmark{1},
Tianlong Chen\IEEEauthorrefmark{3},
Zhiyuan Yu\IEEEauthorrefmark{4},
Neil Zhenqiang Gong\IEEEauthorrefmark{5},
Huan Liu\IEEEauthorrefmark{1}}
\IEEEauthorblockA{\IEEEauthorrefmark{1}Arizona State University,
\IEEEauthorrefmark{2}Stevens Institute of Technology,
\IEEEauthorrefmark{3}University of North Carolina at Chapel Hill,\\
\IEEEauthorrefmark{4}Texas A\&M University,
\IEEEauthorrefmark{5}Duke University\\
Email: \{czhao93, daweili5, huanliu\}@asu.edu, ztan12@stevens.edu, tianlong@cs.unc.edu,\\
zhiyuanyu@tamu.edu, neil.gong@duke.edu}}


%


\maketitle

\begin{abstract}
The rapid advancement of Large Vision-Language Models (LVLMs) is increasingly accompanied by unauthorized scraping and training on multimodal web data, posing copyright and privacy risks to data owners. Existing countermeasures, such as machine unlearning and watermarks, are inherently post-hoc approaches that act only after intellectual property infringement has already occurred. In this work, we propose \MethodName{} to empower data owners to proactively protect their multimodal data against unauthorized LVLM fine-tuning. \MethodName{} generates unlearnable examples by injecting human-imperceptible perturbations that exploit the learning dynamics of LVLMs. By minimizing the training loss, the perturbation creates an optimization shortcut, causing the model to overfit to the noise and thereby degrading downstream performance when the perturbation is absent during inference. To further strengthen this defense, \MethodName{} introduces a cross-modal binding disruption, shifting LVLM attention to enforce a spurious correlation between the noise and the training target with theoretical guarantees. \MethodName{} is evaluated against nine open-source LVLMs across six datasets. Our results demonstrate effective, stealthy, and robust protection under white-box, gray-box, and no-box threat models, establishing a mechanistic advantage in proactively defending against unauthorized fine-tuning exploitation.
\end{abstract}


%
\IEEEpeerreviewmaketitle

\section{Introduction}
\label{sec:introduction}

Large Vision-Language Models (LVLMs)~\cite{singh2025openai,comanici2025gemini} have rapidly become a central component of modern artificial intelligence (AI) systems. By aligning visual content with natural-language instructions and responses, LVLMs support a wide range of multimodal tasks, including visual question answering~\cite{bai2025qwen3}, image captioning~\cite{grattafiori2024llama}, document understanding~\cite{zhu2025internvl3}, and multimodal instruction following~\cite{liu2023visual}. This success is largely fueled by an ever-growing appetite for image-text data. Contemporary LVLMs are pretrained and fine-tuned on hundreds of millions of multimodal pairs scraped indiscriminately from the open web~\cite{schuhmann2022laion,radford2021learning}, which has created an urgent threat for the owners of multimodal content. OpenAI was reported to have transcribed and ingested more than one million hours of YouTube videos without the consent of creators or the platform to obtain multimodal training material for {GPT-4}, with similar practices reported at Google and Meta and a class action subsequently filed by content creators~\cite{metz2024tech,youtuber2024classaction}. Such practices infringe creators' copyrights, leak sensitive information in the content, and deprive owners of control over how content is reused.

\begin{figure}[t]
\centering
\includegraphics[width=0.85\linewidth]{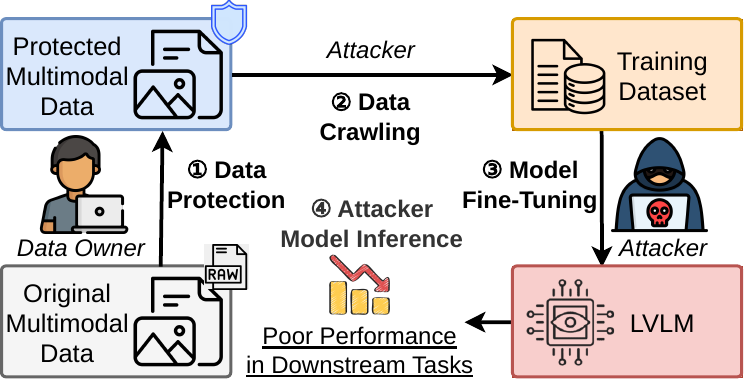}
\caption{Protect multimodal data from unauthorized fine-tuning of LVLMs.}

\label{fig:overview}
\end{figure}

Existing protection mechanisms are insufficient for this setting. Legal takedowns and watermarking~\cite{zhu2018hidden} can help identify or respond to misuse, but they are inherently \emph{post-hoc}, acting only after the data has already been exploited. Model-level remedies such as machine unlearning~\cite{bourtoule2021machine} require cooperation from the model provider and are difficult to supervise externally. Recent efforts have attempted to generate \emph{unlearnable examples}~\cite{huang2021unlearnable,fu2022robust,ren2023transferable} for unauthorized model training by injecting small perturbations. By minimizing the training loss, the perturbation creates an optimized shortcut, causing the model to overfit to the noise and thereby degrading downstream performance when the perturbation is absent during inference. Despite their promise, their effectiveness for LVLM remains unexplored.

Protecting multimodal data against LVLM fine-tuning presents several key challenges. First, input images and text provide redundant yet complementary information that autoregressive LVLMs can effectively exploit to answer a query. Consequently, protecting only one modality leaves residual information, which remains sufficient for effective fine-tuning~\cite{yao2024data}. Second, LVLMs acquire rich semantic knowledge during pretraining. Such prior knowledge enables them to infer the true semantics of samples intended for protection, bypassing defensive perturbations~\cite{li2026when}. Third, defenders typically lack knowledge of the attacker's exact LVLM architecture, data preprocessing pipeline, and fine-tuning procedure, so a practical protection must remain effective in no-box settings. Finally, multimodal protection requires joint optimization across modalities with distinct constraints and optimization dynamics, while considering the asymmetric objectives of the attacker and defender.

Our key insight is that effective multimodal data protection should not only perturb all modalities. \emph{It should also shape how LVLMs capture visual evidence, textual context, and generate target responses during fine-tuning.} By encouraging the LVLMs to rely on a spurious association between the defense signal (i.e., perturbation) and the response, the protection can disrupt the learned multimodal semantics and thereby degrade the generalizability of LVLMs.

In this paper, we propose \MethodName{}, a data-centric framework that proactively protects multimodal data from unauthorized LVLM fine-tuning. Given a clean sample, \MethodName{} jointly protects both modalities by injecting an imperceptible image perturbation and augmenting text with a short trigger. To further strengthen the protection, \MethodName{} introduces two \emph{cross-modal binding disruption mechanisms} that redirect the model learning from genuine semantic associations toward protection-specific spurious correlations with theoretical guarantees. \MethodName{} further optimizes protection across an ensemble of surrogate LVLMs to improve transferability. Finally, these components jointly formulate a multi-objective constrained optimization that balances effectiveness, stealthiness, and robustness.

We conduct comprehensive experiments across six publicly available multimodal datasets and nine open-source LVLMs under white-box, gray-box, and no-box threat models. The results show that \MethodName{} consistently degrades downstream LVLM performance while preserving the perceptual quality of the released samples. We further demonstrate its robustness against adaptive attacks, including data transformations, data purification, adversarial training, and data mixing. Moreover, ablation studies, hyperparameter sensitivity analysis, and attention analysis further validate the contribution of each design component and provide mechanistic insights into how \MethodName{} effectively defends against unauthorized LVLM fine-tuning.

Our study makes the following contributions:
\begin{itemize}[leftmargin=*]
    \item We identify and formalize the problem of safeguarding public multimodal data against unauthorized LVLM fine-tuning, establishing a realistic threat model in which defenders can only modify their data prior to release.
    \item We propose \MethodName{}, to the best of our knowledge the \emph{first data-centric protection framework} that proactively safeguards multimodal data against unauthorized LVLM fine-tuning. \MethodName{} jointly perturbs images and text, crafting coupled perturbations tailored to the autoregressive fine-tuning objective of LVLMs.
    \item We introduce \emph{cross-modal binding disruption}, a mechanism comprising two approaches that redirect the LVLM's learning signal toward planted perturbations and enforce spurious associations between these perturbations and target responses. We further provide a theoretical analysis explaining why this disruption degrades generalizability.
    \item We design an ensemble-based perturbation optimization strategy to improve transferability across unknown LVLM attackers, enabling effective and robust protection under white-box, gray-box, and no-box scenarios.
    \item We conduct a comprehensive evaluation and demonstrate the  effectiveness, transferability, and practicality of \MethodName{} in defending unauthorized LVLM fine-tuning.
\end{itemize}

\section{Related Work}
\label{sec:related}

\noindent\textbf{Defenses Against Unauthorized Data Use.}
Existing defenses against unauthorized data usage in model training either \emph{detect} misuse post-hoc or \emph{prevent} it prior to training. Detection-based defenses audit whether a suspect model was trained on a given piece of data. Some operate without modifying the data, such as membership inference~\cite{shokri2017membership,carlini2022membership} and dataset inference~\cite{maini2021dataset,maini2024llm}; others mark the data before release and later test the suspect model for the mark, such as radioactive data~\cite{sablayrolles2020radioactive}, backdoor-based dataset watermarking~\cite{li2022untargeted}, and provable data-use auditing~\cite{huang2024general}, with extensions to text-to-image diffusion models~\cite{wang2024diagnosis}. While these auditing methods typically require only query access, they are inherently \emph{reactive}: an audit produces evidence of misuse but does not prevent the model from benefiting from the data. Subsequent remediation is either slow (e.g., legal recourse) or contingent on the model provider's cooperation (e.g., machine unlearning~\cite{bourtoule2021machine,liu2025rethinking}). Prevention-based defenses instead \emph{proactively} perturb samples before release so that models trained on them derive minimal utility. Prior work has crafted such protective perturbations against style mimicry~\cite{shan2023glaze,shan2024nightshade}, personalized text-to-image fine-tuning~\cite{van2023anti,liu2024metacloak}, and voice cloning~\cite{yu2023antifake}. These defenses target generative models that synthesize images or audio, whereas LVLMs generate textual answers autoregressively conditioned on both an image and a text input, so protection must address both input modalities and their interaction.

\noindent\textbf{Unlearnable Example Methods.} Unlearnable examples (UEs) provide a proactive, data-centric defense by injecting bounded perturbations that act as \emph{shortcuts}, causing models trained on protected samples to generalize poorly on clean test data. Pioneered for image classification via error-minimizing noise~\cite{huang2021unlearnable}, UEs have since been strengthened for robustness against adversarial training~\cite{fu2022robust,liu2024stable}, cross-architecture transferability~\cite{ren2023transferable,li2026versatile}, and label-agnostic generation via CLIP surrogates~\cite{zhang2023unlearnable}. Subsequent work has extended UEs beyond classification to text~\cite{li2023make}, graphs~\cite{liu2023unlearnable}, speech~\cite{gokul2024poscuda,zhang2023mitigating}, image segmentation~\cite{sun2024unseg}, and diffusion models~\cite{zhao2024unlearnable,li2025styleguard}. In the multimodal domain, MEM~\cite{liu2024multimodal} adapts UEs to image caption alignment task to disrupt CLIP-style contrastive pretraining. Recent studies show that existing UEs can be partially bypassed via subsequent relearning~\cite{dang2023unlearnable}, strong pretrained priors~\cite{li2026priors}, or diffusion purification~\cite{wang2026bridgepure}, underscoring the need for more resilient designs. Despite these advances, no prior work addresses unauthorized fine-tuning of LVLMs, which differs fundamentally from classification and contrastive pretraining in its autoregressive objective, coupled input space, and learning dynamics. \MethodName{} bridges this gap as the first proactive, data-centric protection tailored to LVLM fine-tuning.

\section{Background and Preliminaries}
\label{sec:preliminaries}
We first formulate multimodal data protection problem in Section~\ref{sec:problem_formulation}. We then review the standard formulation of unlearnable examples in Section~\ref{sec:unlearnable_examples}. Later, we introduce the common architecture and fine-tuning paradigm of Large Vision-Language Models in Section~\ref{sec:lvlm}.

\subsection{Problem Formulation}
\label{sec:problem_formulation}
Let $\mathcal{X}$ denote the image space and $\mathcal{T}$ denote natural-language space. A multimodal dataset is given by $\mathcal{D}=\{(x_i,t_i,y_i)\}_{i=1}^{n}$, where $x_i\in\mathcal{X}$ is an image, $t_i\in\mathcal{T}$ is a textual input (e.g., question or instruction), and $y_i\in\mathcal{Y}$ is a target output (e.g., answer or response). A large vision language model is denoted by $f_{\theta}:\mathcal{X}\times\mathcal{T}\rightarrow\mathcal{Y}$, where $\theta=(\theta_{\mathrm{frz}}, \theta_{\mathrm{tr}})$ denotes the model parameters, consisting of frozen pretrained parameters~$\theta_{\mathrm{frz}}$ from the base LVLM checkpoint and trainable parameters~$\theta_{\mathrm{tr}}$ that the attacker may select for fine-tuning. Unlike text-to-image models (e.g., diffusion models~\cite{rombach2022high}), which synthesizes the image based on text, an LVLM take both the image and text as inputs and produce the textual response $y_i$ (e.g., VQA tasks); Given a multimodal training set $\mathcal{D}$, an unauthorized attacker seeks to fine-tune the LVLM to minimize the training loss:

\begin{equation}
    \theta^{\star}(\mathcal{D})
    \in
    \arg\min_{\theta\in\Theta}
    \mathcal{L}_{\mathrm{train}}(\theta;\mathcal{D}),
\end{equation}
where $\mathcal{L}_{\mathrm{train}}$ is the training objective. For simplicity, we use fine-tune and train interchangeably in the paper.

To counter this, the defender releases a protected dataset $\widetilde{\mathcal{D}}_{\phi}=\mathcal{G}_{\phi}(\mathcal{D})$, where $\mathcal{G}_{\phi}$ is a protection map parameterized by $\phi$ and constrained by modality-specific budgets $\widetilde{\mathcal{D}}_{\phi}\in\mathcal{B}(\mathcal{D})$ that keep each sample close to its original. The goal is that models trained on $\widetilde{\mathcal{D}}_{\phi}$ perform poorly on clean downstream data. Given a clean evaluation dataset $\mathcal{D}_{\mathrm{eval}}$ with evaluation loss $\mathcal{L}_{\mathrm{eval}}$, we formalize the defender's objective as the following bilevel optimization problem:
\begin{equation}
\label{eq:problem-formulation}
\begin{aligned}
    \max_{\phi\in\Phi}
    \quad&
    \mathcal{L}_{\mathrm{eval}}
    \bigl(
        \theta^{\star}(\widetilde{\mathcal{D}}_{\phi});
        \mathcal{D}_{\mathrm{eval}}
    \bigr)\\
    \mathrm{s.t.}
    \quad&
    \widetilde{\mathcal{D}}_{\phi}
    =
    \mathcal{G}_{\phi}(\mathcal{D}),
    \quad
    \widetilde{\mathcal{D}}_{\phi}
    \in
    \mathcal{B}(\mathcal{D}),\\
    &
    \theta^{\star}(\widetilde{\mathcal{D}}_{\phi})
    \in
    \arg\min_{\theta\in\Theta}
    \mathcal{L}_{\mathrm{train}}(\theta;\widetilde{\mathcal{D}}_{\phi}).
\end{aligned}
\end{equation}
Essentially, the defender seeks to prevent the unauthorized fine-tuning of LVLMs by degrading the utility of models trained on the protected multimodal dataset.

\subsection{Unlearnable Examples}
\label{sec:unlearnable_examples}
Unlearnable examples are training samples designed to remain semantically useful to human users while inhibiting models from learning meaningful representations~\cite{huang2021unlearnable}. In contrast to traditional adversarial examples, which perturb inputs to cause mistakes at inference time, unlearnable examples intervene before training by modifying the data. The core mechanism is to introduce a subtle perturbation (i.e., shortcut) that models can exploit to reduce the training loss, thereby suppressing them from learning the genuine semantics. Consequently, models trained on such data fit the protected samples but generalize poorly to clean data, offering a promising approach to solve the above problem.

Unlearnable examples are typically generated by error minimization. In the vanilla formulation~\cite{huang2021unlearnable}, given an input $x$ with label $y$, the defender generates error-minimizing noise $\delta$ by solving the \emph{min-min} problem
\begin{equation}
\label{eq:unlearnable-minmin}
    \min_{\theta}\;
    \mathbb{E}_{(x,y)}\Bigl[\min_{\delta}\,\ell\bigl(f_{\theta}(x+\delta),y\bigr)\Bigr]
    \quad\mathrm{s.t.}\quad \|\delta\|_{p}\le\epsilon,
\end{equation}
where $\ell$ is the training loss and $\epsilon$ bounds the perturbation magnitude. The inner minimization (i.e., the defender's perspective) finds noise that makes each sample trivially easy to fit, while the outer minimization (i.e., the attacker's perspective) trains the model on the perturbed data. 

Eq.~\eqref{eq:unlearnable-minmin} can be viewed as a tractable surrogate of our objective in Eq.~\eqref{eq:problem-formulation}. It maximizes the evaluation loss indirectly by minimizing the training loss, which encourages the model to overfit to the shortcut. Yet vanilla UEs are designed for a single continuous input (e.g., an image) associated with a class label, this assumption does not hold for LVLM fine-tuning. Since each sample combines an image and a text input that provide redundant yet complementary evidence for the target response~\cite{yao2024data}, safeguarding only one modality leaves the other as a learnable signal. Effective protection therefore requires jointly optimizing heterogeneous variables: a continuous visual perturbation and a discrete textual modification, each governed by its own budget and optimization dynamics. Coordinating these coupled perturbations presents a key challenge in our work.

\subsection{Large Vision-Language Models}
\label{sec:lvlm}
Modern LVLMs extend pretrained language models with a visual front end that converts images into tokens compatible with the language-model embedding space. Given an image-text input $(x_i,t_i)$, the image is first standardized (e.g., resizing and pixel quantization) by an image processor $g_x$, split into image patches by a patching function $\pi$, and encoded into visual features $H_i^x$ by a visual encoder $E_x$:
\begin{equation}
\label{eq:lvlm-image-features}
    H_i^x = E_x(\pi(g_x(x_i))),
\end{equation}
The visual features are then aligned with the language-model hidden space through a modality projector $P_{x}$:
\begin{equation}
    Z_i^x = P_{x}(H_i^x).
\end{equation}
In parallel, the textual input is tokenized by a tokenizer and is further mapped to embeddings by $E_t$.
\begin{equation}
\label{eq:lvlm-token-construction}
    Z_i^t = E_t(\operatorname{Tokenizer}(t_i))
\end{equation}
The visual and textual tokens are then interleaved into a joint multimodal sequence $S_i$ according to a predefined template:
\begin{equation}
    S_i = \operatorname{Template}(Z_i^x,Z_i^t).
\end{equation}
The language model then decodes the response $y_i$ autoregressively via causal self-attention over $S_i$.

During supervised fine-tuning, LVLM is optimized by minimizing the autoregressive negative log-likelihood of the target response conditioned on the image and textual input:
\begin{equation}
\label{eq:lvlm-finetuning}
    \mathcal{L}_{\mathrm{train}}(\theta;\mathcal{D})
    =
    \frac{1}{n}
    \sum_{i=1}^{n}
    \sum_{j=1}^{|y_i|}
    -\log
    p_{\theta}
    \bigl(
        y_{i,j}
        \mid
        y_{i,<j}, x_i,t_i
    \bigr).
\end{equation}
In contrast to the classifiers targeted by vanilla UEs, which are typically trained from scratch, the attacker here fine-tunes a pretrained LVLM that already encodes rich cross-modal semantics~\cite{li2026when,li2026priors}. Since each response token in Eq.~\eqref{eq:lvlm-finetuning} is predicted from the full multimodal sequence $S_i$, the model can fit protected data relying on such prior knowledge and bypassing the planted shortcut. Consequently, effective protection must also shape how the model uses visual and textual evidence during fine-tuning.

\section{Threat Model}
\label{sec:threat-model}
We consider two parties: an \emph{attacker}, who collects multimodal web data to fine-tune LVLMs, and a \emph{defender}, who owns image-text data and wishes to publish it online while preventing its unauthorized use as effective fine-tuning data.

\noindent\textbf{Attacker Objectives.}
The attacker starts from an existing pretrained LVLM checkpoint and fine-tunes it on multimodal data scraped from the web. Their objective is to maximize the model utility of downstream task while disregarding potential copyright infringement and privacy risks.

\noindent\textbf{Attacker Capabilities.}
The attacker possesses the knowledge and expertise required to optimize LVLMs using a broad range of fine-tuning techniques and to evaluate their performance. In our study, we consider two types of attackers: (\emph{i}) a naive attacker who follows standard fine-tuning pipelines, such as LoRA, and (\emph{ii}) an adaptive attacker who is aware of \MethodName{} and seeks to circumvent its protection. The adaptive attacker may employ diverse strategies, including data transformations, data purification, adversarial attack and data mixing. We discuss these strategies in Section~\ref{sec:adaptive}.

\noindent\textbf{Defense Objectives.}
The defender may be an individual artist, news outlet, private dataset curator, or any data owner who seeks to prevent their content from being effectively exploited by LVLMs. The defender has two primary objectives. First, protected samples should disrupt unauthorized fine-tuning: when incorporated into the attacker's training set, they should degrade the model's downstream performance. Second, the protected content should preserve high perceptual and semantic fidelity, ensuring that it remains useful for legitimate online publication and human consumption.

\noindent\textbf{Defense Assumptions.}
The defender has full access to their own image-text pairs before publication and can modify them by adding imperceptible protective perturbations. However, the defender has no access to the attacker's full training corpus, fine-tuning recipe, hyperparameters, or model weights, and cannot directly interfere with the attacker's training process. We consider three levels of defender knowledge about the attacker's target model. In the \emph{white-box} setting, the defender knows the exact LVLM that the attacker will fine-tune and optimizes the protection directly against it. This setting provides an upper bound on defense effectiveness, although it is rarely realistic in practice. In the \emph{gray-box} setting, the defender knows only partial information, such as the model family or architecture, but not the exact checkpoint. In the \emph{no-box} setting, the defender has no concrete knowledge of the attacker's model and must rely on architecture-agnostic protection and cross-family transferability, which is the most realistic scenario.

\section{Our Approach \MethodName{}}
\label{sec:method}

We first construct a continuous image perturbation in Section~\ref{sec:image-protection} and a discrete textual trigger in Section~\ref{sec:text-protection} to form a coupled multimodal protection that eliminates residual learnable signal in either modality. We then introduce cross-modal binding disruption in Section~\ref{sec:binding-disruption} to redirect LVLM learning toward protection-specific shortcuts, preventing the autoregressive objective from bypassing the protection through pretrained knowledge. We further optimize the protection over an ensemble of surrogate LVLMs in Section~\ref{sec:ensemble-protection} to improve the transferability of protection to no-box scenarios. We finally cast the framework as a constrained multi-objective bilevel optimization in Section~\ref{sec:joint-optimization} to reconcile heterogeneous variables under modality-specific budgets and the asymmetric objectives of defender and attacker and provide tractable solution in Algorithm~\ref{alg:mmguard-optimization}.

\begin{figure*}[t]
\centering
\includegraphics[width=1\linewidth]{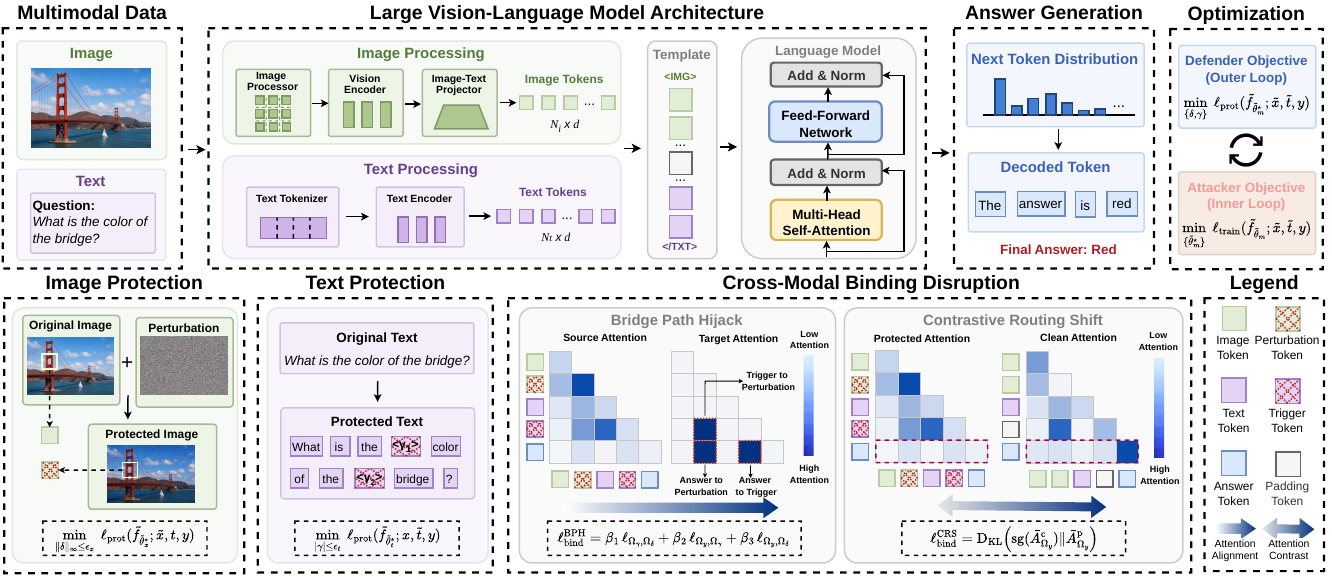}
\caption{Overview of \MethodName{}. \emph{Top:} an image-text sample passes through the LVLM (image and text processing, template, and language model) to generate the answer. \emph{Bottom:} \MethodName{} injects a bounded perturbation into the image, inserts text triggers, and applies cross-modal binding disruption, where Bridge Path Hijack (BPH) steers attention onto trigger and perturbation tokens and Contrastive Routing Shift (CRS) pushes protected attention away from clean attention. Finally, the defender optimizes the protection (outer loop) against the attacker's fine-tuning (inner loop) via bilevel optimization.}
\label{fig:framework}
\end{figure*}

\subsection{Unlearnable Image Protection}
\label{sec:image-protection}
We formulate the unlearnable image protection as two key components: a differentiable image processor that allows gradient flow from the surrogate LVLM back to pixel-space perturbations, and projected gradient descent optimization that updates perturbations to protect the image.

\noindent\textbf{Differentiable Image Processor.} As described in the Section~\ref{sec:lvlm}, the image processor $g_x$ contains non-differentiable operations that create an obstacle for directly optimizing the image perturbation $\delta_i$ in the raw pixel space. This obstacle undermines the defender's ability to create effective, stealthy, unlearnable examples. To address it, we construct a differentiable surrogate processor $\tilde{g}_x$ that approximates the practical processor $g_x$: $\tilde{g}_x\approx g_x$.
Specifically, we handle the non-differentiable functions by three strategies. (\emph{i})~For discretization steps (e.g., pixel quantization), we keep the forward value used by the practical processor but copy gradients through a continuous proxy. (\emph{ii})~For geometric transformations (e.g., resizing), we use differentiable substitutes that approximate the same transformation during backpropagation. (\emph{iii})~For layout decisions (e.g., resolution selection), we follow the same deterministic rules as $g_x$ and treat the resulting layout as fixed metadata for the current image. Operations that are already differentiable are kept unchanged. Together, these choices enable gradient backpropagation to pixel-space perturbations while preserving the surrogate model’s forward representation.

\noindent\textbf{Projected Gradient Descent Optimization.}
For each clean image $x_i$, we construct the protected image $\tilde{x}_i$ by adding a human-imperceptible perturbation $\delta_i$: $\tilde{x}_i = x_i+\delta_i$, where $\delta_i$ is constrained to lie in the image-side feasible set $\mathcal{C}^{\delta}_{i}$:
\begin{equation}
\label{eq:image_constrain}
    \mathcal{C}^{\delta}_{i}
    =
    \{\delta:\; x_i+\delta\in\mathcal{X},\ \|\delta\|_{\infty}\le\epsilon_x\}.
\end{equation}

Formally, the unlearnable image protection solves the following bilevel optimization problem:
\begin{align}
\min_{\{\delta_i\}}
&\quad
\frac{1}{n}
\sum_{i=1}^{n}
\ell_{\mathrm{prot}}
\bigl(
\tilde{f}_{\tilde{\theta}_x^\star};
\tilde{x}_i,t_i,
y_i
\bigr)
\nonumber\\
\mathrm{s.t.}
&\quad
\tilde{x}_i = x_i + \delta_i,\quad \delta_i\in\mathcal{C}^{\delta}_{i},
\nonumber\\
&\quad
\tilde{\theta}_x^\star
\in
\arg\min_{\tilde{\theta}}
\sum_{i=1}^{n}
\ell_{\mathrm{train}}
\bigl(
\tilde{f}_{\tilde{\theta}};
\tilde{x}_i,t_i,
y_i
\bigr).
\label{eq:image-only-objective}
\end{align}
where $\ell_{\mathrm{prot}}$ is the defender's protection objective (e.g., the LVLM training loss), later defined in Eq.~\eqref{eq:binding-loss}, and $\ell_{\mathrm{train}}$ denotes the sample-wise LVLM training loss.

We optimize $\delta_i$ by projected gradient descent:
\begin{equation}
\delta_i \leftarrow
\Pi_{\mathcal{C}^{\delta}_{i}}
\!\left(
\delta_i - \alpha_x\, \mathrm{sign}
\bigl(\nabla_{\delta_i}\ell_{\mathrm{prot}}\bigr)
\right).
\label{eq:image-pgd}
\end{equation}
where $\Pi_{\mathcal{C}^{\delta}_{i}}$ denotes projection back onto the per-image feasible set, and $\alpha_x$ is the step size. The gradient is taken through the differentiable processor $\tilde{g}_x$, the visual 
encoder $E_x$, the projector $P_x$, and the language model.

\subsection{Unlearnable Text Protection}
\label{sec:text-protection}

Given a clean text input $t_i$, we construct the protected text $\tilde{t}_i$ by inserting a bounded-length trigger $\gamma_i=(\gamma_{i,1},\ldots,\gamma_{i,|\gamma_i|})$: $\tilde{t}_i = \operatorname{Insert}(t_i,\gamma_i)$, where insertion is restricted to textual input positions and does not replace tokens in $t_i$. Each trigger token is selected from an admissible vocabulary $\mathcal{V}_{\mathrm{adm}}\subseteq\mathcal{V}$ (e.g., excluding control tokens, and non-linguistic tokens). The trigger length is bounded by the text budget length $\epsilon_t$, so the text-side feasible set for $\gamma_i$ is:
\begin{equation}
\label{eq:text_constrain}
    \mathcal{C}^{\gamma}_{i}
    =
    \{\gamma:\; |\gamma|\le \epsilon_t,\ \gamma_j\in\mathcal{V}_{\mathrm{adm}}\}.
\end{equation}

Analogous to image-side protection, the unlearnable text protection solves a discrete \emph{min-min} problem:
\begin{align}
\min_{\{\gamma_i\}}
&\quad
\frac{1}{n}
\sum_{i=1}^{n}
\ell_{\mathrm{prot}}
\bigl(
\tilde{f}_{\tilde{\theta}_t^\star};
x_i,\tilde{t}_i,
y_i
\bigr)
\nonumber\\
\mathrm{s.t.}
&\quad
\tilde{t}_i=\operatorname{Insert}(t_i,\gamma_i),
\quad
\gamma_i\in\mathcal{C}^{\gamma}_{i},
\nonumber\\
&\quad
\tilde{\theta}_t^\star
\in
\arg\min_{\tilde{\theta}}
\sum_{i=1}^{n}
\ell_{\mathrm{train}}
\bigl(
\tilde{f}_{\tilde{\theta}};
x_i,\tilde{t}_i,
y_i
\bigr).
\label{eq:text-only-objective}
\end{align}

\noindent\textbf{Gradient-Based Candidate Insertion.}
We approximate the outer update in Eq.~\eqref{eq:text-only-objective} with a HotFlip-style~\cite{ebrahimi2018hotflip} first-order search over the inserted trigger tokens. Let $E_t$ map vocabulary tokens to embeddings and denote the embedding of token $v$ by $e_v=E_t(v)$. For a trigger position $j$, let $\gamma_i^{[j=v]}$ denote the trigger obtained by setting that inserted position to candidate token $v$. A first-order expansion of the protection loss around the current trigger embedding gives
\begin{equation}
\ell_{\mathrm{prot}}\!\bigl(\gamma_i^{[j=v]}\bigr)
\approx
\ell_{\mathrm{prot}}(\gamma_i)
+
\bigl(e_v-e_{\gamma_{i,j}}\bigr)^{\top}
\nabla_{e_{\gamma_{i,j}}}\ell_{\mathrm{prot}} .
\end{equation}
The corresponding linear score is
\begin{equation}
s_{i,j}(v)
=
\bigl(e_v-e_{\gamma_{i,j}}\bigr)^{\top}
\nabla_{e_{\gamma_{i,j}}}\ell_{\mathrm{prot}} .
\label{eq:hotflip-score}
\end{equation}

\noindent\textbf{Candidate Verification.}
The token with the best first-order score may still be suboptimal since the HotFlip score is computed in the token-embedding space, whereas the released trigger is ultimately a surface string processed by the LVLM tokenizer. Considering byte-pair encoding (BPE) tokenization of LVLM, decoding a candidate token and inserting it into the text can change neighboring token boundaries through merge or split operations, so the actual re-tokenized sequence may differ from the assumed single-position substitution. To account for this tokenization mismatch, we verify the shortlisted candidates and computing the exact protection loss through a forward pass of the LVLM:

\begin{equation}
v_{i,j}^{\star}
=
\arg\min_{v\in\mathcal{V}^{\mathrm{cand}}_{i,j}}
\ell_{\mathrm{prot}}\!\bigl(\gamma_i^{[j=v]}\bigr).
\label{eq:hotflip-verify}
\end{equation}
This screen-and-verify procedure is provably near-optimal: under standard smoothness, the selected token matches the best admissible substitution up to a small quadratic remainder, which is justified by Lemma~\ref{lem:hotflip-screening} in Appendix~\ref{app:hotflip-screening-proof}.

\subsection{Cross-Modal Binding Disruption}
\label{sec:binding-disruption}
Image and text unlearnable protections discourage an LVLM from learning useful features from individual modalities, but they do not necessarily prevent the model from learning generalizable multimodal associations. In LVLMs, answer tokens are generated by decoder attention over the full multimodal context. Even when both the image and the text are protected, this attention mechanism can still route generation through semantically meaningful evidence that remains available in either modality. Moreover, LVLMs have prior knowledge of multimodal semantics from pretraining. LVLMs can rely on existing associations to exploit protected data, thereby undermining the protection. Therefore, potent protection should not be limited to surface-level perturbations in both modalities; it also shapes how LVLMs capture visual evidence and textual context and how they generate target responses during fine-tuning. To enable this, \MethodName{} introduces two \emph{cross-modal binding disruption mechanisms}. By disrupting semantic binding and steering optimization toward protection-specific spurious attention paths, the method encourages fine-tuning to rely on non-transferable perturbation shortcuts.

\noindent\textbf{Attention Mass Distribution.}
For a protected sample $(\tilde x_i,\tilde t_i,y_i)$, the LVLM template assembles a joint token sequence in the language model backbone:
\begin{equation}
    \Omega_i
    =
    \Omega_{x,i}
    \uplus
    \Omega_{t,i}
    \uplus
    \Omega_{\gamma,i}
    \uplus
    \Omega_{y,i},
\label{eq:token-union}
\end{equation}
where $\Omega_{x,i}$, $\Omega_{t,i}$, $\Omega_{\gamma,i}$, and $\Omega_{y,i}$ denote the image tokens, original text tokens, inserted trigger tokens, and answer tokens, respectively. For simplicity, we omit auxiliary tokens.

Among the image tokens, we further identify the subset most affected by the perturbation. Let $\mathcal{P}_{i,b}$ denote the preprocessed pixel region that corresponds to image token $b\in\Omega_{x,i}$, and define its average perturbation magnitude as
\begin{equation}
    \rho_{i,b}
    =
    \frac{1}{|\mathcal{P}_{i,b}|}
    \sum_{u\in\mathcal{P}_{i,b}}
    |\delta_i(u)|.
\end{equation}
We select perturbation tokens $\Omega_{\delta,i}$ that are most affected in the image tokens based on a ratio $\tau_{\delta}\in(0,1]$:
\begin{equation}
    \Omega_{\delta,i}
    =
    \left\{
    b\in\Omega_{x,i}:\;
    \operatorname{rank}_{\Omega_{x,i}}(\rho_{i,b})
    \le
    \left\lceil
    \tau_{\delta}|\Omega_{x,i}|
    \right\rceil
    \right\},
\label{eq:perturbation-token-set}
\end{equation}
where $\operatorname{rank}_{\Omega_{x,i}}(\cdot)$ orders perturbation scores in descending order, with ties broken deterministically. 

Let $A^{(k,h)}\!\in\![0,1]^{|\Omega|\times|\Omega|}$ denote the attention matrix of head $h$ in layer $k$, where $h\in\mathcal{H}$ and each row is a distribution over key tokens. For a source token set $\mathcal{S}\subseteq\Omega$, we define its head-averaged attention-mass distribution as
\begin{equation}
    \bar{A}_{k}(\mathcal{S})
    =
    \frac{1}{|\mathcal{H}|\,|\mathcal{S}|}
    \sum_{h\in\mathcal{H}}
    \sum_{a\in\mathcal{S}}
    A^{(k,h)}_{a,:}
    \in \Delta^{|\Omega|-1}.
\label{eq:attention-mass}
\end{equation}
This distribution summarizes how much attention mass the queries in $\mathcal{S}$ collectively assign to every key token. For a nonempty target token set $\mathcal{R}\subseteq\Omega$, define the uniform distribution on $\mathcal{R}$, embedded in the simplex over $\Omega$, as
\begin{equation}
    U_{\mathcal{R}}
    =
    \frac{1}{|\mathcal{R}|}\mathbf{1}_{\mathcal{R}}
    \in\Delta^{|\Omega|-1},
    \qquad
    \mathbf{1}_{\mathcal{R}}(b)=
    \begin{cases}
        1, & b\in\mathcal{R},\\
        0, & b\notin\mathcal{R}.
    \end{cases}
\label{eq:uniform-token-distribution}
\end{equation}
That is, $U_{\mathcal{R}}$ assigns equal probability to tokens in the target set and zero probability to all other tokens. We then measure how far the source attention mass is from the reference distribution, averaged over layers in $\mathcal{K}$:
\begin{equation}
    \ell_{\mathrm{mass}}(\mathcal{S},\mathcal{R})
    =
    \frac{1}{|\mathcal{K}|}
    \sum_{k\in\mathcal{K}}
    \mathrm{D}_{\mathrm{KL}}\!\bigl(U_{\mathcal{R}}\,\|\,\bar{A}_{k}(\mathcal{S})\bigr).
\label{eq:routing-cost}
\end{equation}
Equivalently, $\ell_{\mathrm{mass}}$ is a layer-averaged contrastive loss that raises the attention logits of tokens in $\mathcal{R}$ and suppresses those of other tokens (Proposition~\ref{prop:routing-contrastive}, Appendix~\ref{app:routing-contrastive-proof}).

\noindent\textbf{Variant 1: Bridge Path Hijack (BPH).}
\label{sec:bph}
Based on the attention-mass loss, we can design various cross-modal binding objectives that create different protection-specific attention shortcuts. Specifically, we consider three distinct paths:

\begin{itemize}[leftmargin=*]
    \item \emph{Trigger-to-perturbation binding} $\ell_{\mathrm{mass}}(\Omega_{\gamma},\Omega_{\delta})$: it ties the text trigger to the perturbation, so that the image and text protections are coupled across modalities.
    \item \emph{Answer-to-trigger shortcut} $\ell_{\mathrm{mass}}(\Omega_y,\Omega_{\gamma})$: it creates a direct attention path from answer tokens to the inserted trigger, so that the answer loss can be reduced by attending to the trigger rather than to the original semantic evidence.
    \item \emph{Answer-to-perturbation shortcut} $\ell_{\mathrm{mass}}(\Omega_y,\Omega_{\delta})$: it preserves a spurious attention shortcut from answer tokens to perturbation that degrades downstream generalization.
\end{itemize}

We combine these three paths to form the BPH loss:

\begin{align}
    \ell_{\mathrm{bind}}^{\mathrm{BPH}}
    &=
    \beta_{1}\,\ell_{\mathrm{mass}}(\Omega_{\gamma},\Omega_{\delta})
    +\beta_{2}\,\ell_{\mathrm{mass}}(\Omega_y,\Omega_{\gamma})
    \nonumber\\
    &+\beta_{3}\,\ell_{\mathrm{mass}}(\Omega_y,\Omega_{\delta}),
\label{eq:bph}
\end{align}
where $\beta_{1},\beta_{2},\beta_{3}\ge 0$ are hyperparameters that control the relative importance of the three paths. 

\begin{theorem}[Effectiveness of BPH]
\label{thm:effectivenss-BPH}
Let $\mathcal{R}_{i}\in\{\Omega_{\gamma,i},\,\Omega_{\delta,i}\}$ be a protection-induced target set that reduces to the empty set under the clean evaluation input $(x_{i},t_{i},y_{i})$. Suppose a model $\theta^{\star}$ achieves answer-token binding $\ell_{\mathrm{mass}}(\Omega_{y,i},\mathcal{R}_{i})\le\eta$ on the protected sample $(\tilde{x}_{i},\tilde{t}_{i},y_{i})$. Denote by $\bar{A}^{\mathrm{p}}_{k,i}(\Omega_y)$ and $\bar{A}^{\mathrm{c}}_{k,i}(\Omega_y)$ the head-averaged answer-token attention-mass distributions of $\theta^{\star}$ on the protected and clean inputs at layer $k$, both extended to a common token universe by zero-padding the positions in $\mathcal{R}_{i}$ on the clean side. Then
\begin{equation}
    \frac{1}{|\mathcal{K}|}\sum_{k\in\mathcal{K}}
    \mathrm{TV}\!\bigl(
        \bar{A}^{\mathrm{p}}_{k,i}(\Omega_y),\,
        \bar{A}^{\mathrm{c}}_{k,i}(\Omega_y)
    \bigr)
    \;\ge\; 1-\sqrt{\eta/2}.
\label{eq:tv-shift-bound}
\end{equation}
The proof is given in Appendix~\ref{app:shortcut-gap-proof}.
\end{theorem}
Theorem~\ref{thm:effectivenss-BPH} shows that every protection-induced binding term in Eq.~\eqref{eq:bph} that is driven to a small $\eta$ enforces an attention-pattern shift of at least $1-\sqrt{\eta/2}$ in mean total-variation distance between protected training and clean evaluation, providing the mechanistic effectiveness by which BPH's protection signal fails to transfer to clean downstream data. The complementary trigger-to-perturbation term ensures that the trigger and perturbation pathways are coupled: severing either at evaluation collapses the joint route the surrogate adopted during fine-tuning to strengthen protection.

\noindent\textbf{Variant 2: Contrastive Routing Shift (CRS).}
\label{sec:crs}
A fixed bridge is effective when the surrogate and the attacker share a similar fusion mechanism, but its prescriptive form may over-constrain the path and limit transfer to unseen architectures. CRS relaxes the prescription: instead of dictating \emph{where} the answer should attend, it only requires that the answer's attention-mass pattern \emph{differs} between the clean and protected forward passes. Let $\bar{A}_{k}^{\mathrm{c}}(\Omega_y)$ and $\bar{A}_{k}^{\mathrm{p}}(\Omega_y)$ denote the answer-token attention-mass distributions obtained under the clean and protected inputs, respectively. CRS minimizes
\begin{equation}
    \ell_{\mathrm{bind}}^{\mathrm{CRS}}
    =
    -\frac{1}{|\mathcal{K}|}
    \sum_{k\in\mathcal{K}}
    \mathrm{D}_{\mathrm{KL}}\!\bigl(
        \mathrm{sg}(\bar{A}_{k}^{\mathrm{c}}(\Omega_y))
        \,\|\,
        \bar{A}_{k}^{\mathrm{p}}(\Omega_y)
    \bigr),
\label{eq:crs}
\end{equation}
where the stop-gradient $\mathrm{sg}(\cdot)$ freezes the clean reference so that updates only reshape the protected-side route. CRS does not assume a specific target that the surrogate must attend to, thereby creating another effective protection variant.

\noindent\textbf{Joint Protection Objective.} For each cross-modal binding disruption loss, we combine it with the standard training loss to form a joint protection objective:
\begin{equation}
    \ell_{\mathrm{joint}}
    =
    \lambda_{\mathrm{train}}\,\ell_{\mathrm{train}}
    +
    \lambda_{\mathrm{bind}}\,\ell_{\mathrm{bind}}.
\label{eq:binding-loss}
\end{equation}
The first term retains the standard min-min unlearnable-example signal: the protected sample must remain easy to fit during unauthorized fine-tuning. The second term governs \emph{how} that fitting is achieved, biasing the surrogate toward attention routes that exist only when the trigger is inserted and the image perturbation is present. The two terms are complementary, thus forming more effective protection.

\subsection{Ensemble Protection Variant}
\label{sec:ensemble-protection}
The joint objective in Eq.~\eqref{eq:binding-loss} optimizes protection with respect to a single surrogate LVLM, which may provide limited robustness against unseen attackers with different model architectures. To improve transferability, we extend the protection objective to an ensemble of surrogate models:

\begin{equation}
    \ell_{\mathrm{prot}} =
    \sum_{m=1}^{M}
    \omega_m \ell_{\mathrm{joint}}^{(m)},
    \qquad
    \sum_{m=1}^{M}\omega_m=1,
\label{eq:ensemble-loss}
\end{equation}
where $\ell_{\mathrm{train}}^{(m)}$, $\ell_{\mathrm{bind}}^{(m)}$ are the training loss and the binding disruption loss computed on surrogate $m$, respectively, and $\omega_m\ge 0$ denotes the weight for surrogate $m$. 

\subsection{Adversarial Protection as a Variant}
\label{sec:adversarial-objective}
The ensemble objective in Eq.~\eqref{eq:ensemble-loss} follows the standard \emph{min-min} formulation, where protected samples are optimized to induce the unlearnable protection. This versatile framework can be further adapted into an adversarial protection variant, which protects multimodal data by disrupting LVLM fine-tuning through maximization of the training-loss term:

\begin{equation}
\begin{aligned}
    \ell_{\mathrm{prot}}^{\mathrm{AD}}
    &=
    \sum_{m=1}^{M}
    \omega_m\,\ell_{\mathrm{joint}}^{\mathrm{AD} (m)},
    \\
    \ell_{\mathrm{joint}}^{\mathrm{AD} (m)}
    &=
    -\lambda_{\mathrm{train}}\,
    \ell_{\mathrm{train}}^{(m)}
    +
    \lambda_{\mathrm{bind}}\,
    \ell_{\mathrm{bind}}^{(m)} .
\end{aligned}
\label{eq:adversarial-loss}
\end{equation}
This adversarial variant targets training disruption by making protected samples hard to fit, thus directly degrading fine-tuning performance rather than relying on shortcut memorization. The binding term is preserved with the same sign because it shifts attention routing patterns that are shared across LVLMs to amplify the protection effect.

\subsection{Constrained Multi-objective Optimization}
\label{sec:joint-optimization}
Given a clean sample $(x_i,t_i,y_i)$, the unlearnable image and text protections create two specialized protection signals, and the cross-modal binding disruption creates a complementary attention-level mechanism guidance. We combine them into a single constrained bilevel optimization over the released dataset, where the protected input is
\begin{equation}
    \tilde{x}_i=x_i+\delta_i,
    \qquad
    \tilde{t}_i=\operatorname{Insert}(t_i,\gamma_i),
\end{equation}
with $\delta_i\in\mathcal{C}^{\delta}_{i}$ and $\gamma_i\in\mathcal{C}^{\gamma}_{i}$. Given a surrogate model set $\{\tilde f^{(m)}\}_{m=1}^{M}$, \MethodName{} solves the follow problem
\begin{align}
\min_{\{\delta_i,\gamma_i\}}
&\quad
\frac{1}{n}
\sum_{i=1}^{n}
\sum_{m=1}^{M}
\omega_m
\ell_{\mathrm{joint}}^{(m)}
\bigl(
\tilde f_{\tilde{\theta}^{(m)\star}};
\tilde{x}_i,\tilde{t}_i,y_i
\bigr)
\nonumber\\
\mathrm{s.t.}
&\quad
\delta_i\in\mathcal{C}^{\delta}_{i},
\quad
\gamma_i\in\mathcal{C}^{\gamma}_{i},
\quad \forall i,
\quad
\sum_{m=1}^{M}\omega_m=1,
\nonumber\\
&\quad
\tilde{\theta}^{(m)\star}
\in
\arg\min_{\tilde{\theta}^{(m)}}
\sum_{i=1}^{n}
\ell_{\mathrm{train}}
\bigl(
\tilde f_{\tilde{\theta}^{(m)}};
\tilde{x}_i,\tilde{t}_i,y_i
\bigr),
\quad
\forall m .
\label{eq:joint-constrained-objective}
\end{align}
Here $\ell_{\mathrm{joint}}^{(m)}$ is instantiated by Eq.~\eqref{eq:binding-loss} for the min-min unlearnable objective, or by Eq.~\eqref{eq:adversarial-loss} for the adversarial variant. This formulation is multi-objective in two senses. First, the defender optimizes both the task-level fitting signal and the attention-level binding signal through $\lambda_{\mathrm{train}}$ and $\lambda_{\mathrm{bind}}$. Second, the ensemble weights $\omega_m$ approximate transfer to possible attackers by optimizing the same released perturbations against multiple surrogates. The modality budgets remain hard constraints, ensuring that the optimized data stays within the protection set in Eq.~\eqref{eq:image_constrain} and Eq.~\eqref{eq:text_constrain}.

\begin{algorithm}[t]
\caption{\MethodName{}: multimodal data protection via constrained multi-objective optimization.}
\label{alg:mmguard-optimization}
\KwIn{Clean data $\mathcal{D}=\{(x_i,t_i,y_i)\}_{i=1}^{n}$; surrogate LVLMs $\{\tilde f^{(m)}\}_{m=1}^{M}$; budgets $(\epsilon_x,\epsilon_t)$; outer rounds $R$; inner steps $Q$.}
\KwOut{Protected dataset $\widetilde{\mathcal{D}}=\{(\tilde{x}_i,\tilde{t}_i,y_i)\}_{i=1}^{n}$.}
\BlankLine
Initialize feasible $\delta_i\in\mathcal{C}^{\delta}_{i}$ and $\gamma_i\in\mathcal{C}^{\gamma}_{i}$ for all $i$\;
\For{$r\leftarrow 1$ \KwTo $R$}{
    Form $\tilde{x}_i=x_i+\delta_i$ and $\tilde{t}_i=\operatorname{Insert}(t_i,\gamma_i)$\;
    \For{each surrogate $m=1,\ldots,M$}{
        Approximate $\tilde{\theta}^{(m)\star}$ by $Q$ gradient steps on $\sum_i\ell_{\mathrm{train}}\bigl(\tilde f_{\tilde{\theta}^{(m)}};\tilde{x}_i,\tilde{t}_i,y_i\bigr)$\;
    }
    Compute the ensemble protection loss in Eq.~\eqref{eq:joint-constrained-objective}\;
    Update each $\delta_i$ by PGD and project onto $\mathcal{C}^{\delta}_{i}$\;
    Update each $\gamma_i$ by gradient-guided candidate screening and exact candidate verification over $\mathcal{V}_{\mathrm{adm}}$\;
}
\Return $\widetilde{\mathcal{D}}=\bigl\{\bigl(\tilde{x}_i,\tilde{t}_i,y_i\bigr)\bigr\}_{i=1}^{n}$\;
\end{algorithm}

\noindent\textbf{Optimization Algorithm.}
Directly solving Eq.~\eqref{eq:joint-constrained-objective} is impractical because it would require repeatedly solving the attacker's inner fine-tuning problem to convergence, while jointly optimizing continuous image perturbations and discrete text triggers under different feasibility constraints. We therefore use an alternating approximation summarized in Algorithm~\ref{alg:mmguard-optimization}. At each outer round, the current protected samples are used to approximate the inner solution by a small number of adaptation steps for all surrogate models. The outer loss is then evaluated on the adapted surrogate. Image perturbations are updated by projected gradient descent as in Eq.~\eqref{eq:image-pgd}, while trigger tokens are updated by the HotFlip screening and verification procedure in Eqs.~\eqref{eq:hotflip-score}-\eqref{eq:hotflip-verify}. These two updates share the same joint loss, so the image perturbation and text trigger are optimized to support the same protection purpose. The algorithm preserves feasibility after every update: projection clips image perturbations to the $\ell_{\infty}$ budget and the valid pixel range, while text updates only select admissible tokens within the fixed trigger-length budget. This algorithm balances overall feasibility with the effectiveness, stealthiness, and robustness of the protection.

\section{Experiments}

We empirically validate \MethodName{} along six axes. We first detail the datasets, target LVLMs, attack scenarios, and evaluation metrics in Section~\ref{sec:experimental-setup}. We then quantify protection effectiveness under white-box and gray-box attackers in Section~\ref{sec:effectiveness}, and assess transferability to no-box LVLMs and diverse fine-tuning techniques in Section~\ref{sec:transferability}. We further evaluate robustness against adaptive attackers that apply adaptive attacks by data transformations, data purification, adversarial training, and data mixing in Section~\ref{sec:adaptive}, and measure practicality in terms of stealthiness and computational cost in Section~\ref{sec:practicality}. We finally analyze the protection mechanism in Section~\ref{sec:mechanism} through ablations, parameter sensitivity, and attention-level visualization.

\subsection{Experimental Setup}
\label{sec:experimental-setup}
\noindent\textbf{Datasets and Tasks.}
We evaluate \MethodName{} across six public multimodal question-answering benchmarks spanning general visual reasoning, domain-specific reasoning, text-centric perception, and document understanding. RealWorldQA~\cite{realworldqa2024} tests understanding of everyday physical scenes, including spatial relations, object affordances, and commonsense visual cues. MMStar~\cite{chen2024we} evaluates fine-grained multimodal reasoning across six capability categories and eighteen task axes. ScienceQA~\cite{lu2022learn} covers curriculum-grounded science reasoning that combines visual evidence with textual context. VQA-RAD~\cite{lau2018dataset} evaluates radiology-oriented medical VQA over clinical images. TextVQA~\cite{singh2019towards} measures scene-text reading and reasoning in natural images, while DocVQA~\cite{mathew2021docvqa} targets document image understanding over structured and semi-structured text. Table~\ref{tab:dataset-statistics} summarizes the dataset statistics.

\begin{table}[t]
\centering
\caption{Statistics of the evaluation datasets.}
\label{tab:dataset-statistics}
\resizebox{\linewidth}{!}{%
\begin{tabular}{@{}llll@{}}
\toprule
Dataset & Domain & Task & Size \\
\midrule
RealWorldQA & General & Visual Understanding & 765 \\
MMStar & General & Multimodal Reasoning & 1,500 \\
ScienceQA & Science & Science Reasoning & 21,208 \\
VQA-RAD & Medical & Medical VQA & 2,244 \\
TextVQA & Scene text & OCR Reasoning & 45,336 \\
DocVQA & Document & Document Understanding & 50,000 \\
\bottomrule
\end{tabular}
}
\end{table}

\noindent\textbf{Large Vision-Language Models and Training Configuration.}
We consider and evaluate state-of-the-art open-source LVLMs with various architectures and parameter scales including Qwen3-VL-2B-Instruct, Qwen3-VL-4B-Instruct, and Qwen3-VL-8B-Instruct~\cite{bai2025qwen3}; Qwen2.5-VL-7B-Instruct~\cite{Bai2025Qwen25VLTR}; Llama-3.2-11B-Vision-Instruct~\cite{grattafiori2024llama}; LLaVA-v1.5-7B~\cite{liu2024improved}; InternVL3.5-8B~\cite{wang2025internvl35}; Gemma-3-4B-IT~\cite{team2025gemma}; and GLM-4.1V-9B-Base~\cite{hong2025glm}. We consider three attack scenarios with different levels of surrogate knowledge: (\emph{i}) white-box: we use Qwen3-VL-4B-Instruct as the surrogate model and evaluate it as the white-box attacker; (\emph{ii}) gray-box: we use Qwen3-VL-4B-Instruct as the surrogate and evaluate three Qwen-family variants with different scales or generations, i.e., Qwen3-VL-2B-Instruct, Qwen3-VL-8B-Instruct, and Qwen2.5-VL-7B-Instruct; and (\emph{iii}) no-box: we leverage Qwen3-VL-4B-Instruct~\cite{bai2025qwen3} and MiniCPM-V-4~\cite{yao2025efficient} as surrogate models to generate the protected data and evaluate the other five LVLMs as no-box attackers.

\noindent\textbf{Evaluation Metrics.}
Our primary utility metric is task accuracy on the clean held-out test set after unauthorized fine-tuning. To quantify protection effectiveness, we further report the accuracy drop relative to clean fine-tuning:
\begin{equation}
\Delta_{\mathrm{acc}} =
\mathrm{Acc}\bigl(f_{\theta_{\mathrm{clean}}};\mathcal{D}_{\mathrm{test}}\bigr)
-
\mathrm{Acc}\bigl(f_{\theta_{\mathrm{prot}}};\mathcal{D}_{\mathrm{test}}\bigr),
\end{equation}
Larger $\Delta_{\mathrm{acc}}$ indicates stronger protection, while clean fine-tuning accuracy serves as the task-specific upper reference. We provide complete dataset descriptions, fine-tuning and protection-optimization configurations, and the per-dataset evaluation protocol in Appendices~\ref{app:dataset-details} and~\ref{app:training-details}.

\subsection{Effectiveness Evaluation}
\label{sec:effectiveness}
We assess whether \MethodName{} produces a reliable protection signal when the attacker fine-tunes on the released data. We examine (\emph{i})~effectiveness under white-box and gray-box scenarios and (\emph{ii})~attacker training dynamics.

\begin{figure}[t]
\centering
\includegraphics[width=\linewidth]{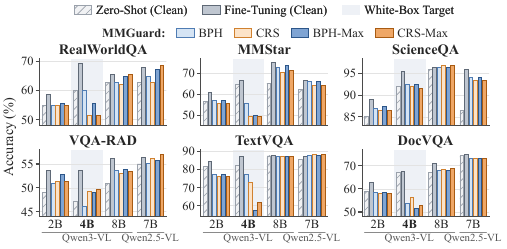}
\caption{Protection effectiveness across \MethodName{} variants under white-box and gray-box scenarios. Lower bars indicate stronger protection.}
\label{fig:effectiveness}
\end{figure}

\noindent\textbf{Protection Effectiveness Under White-Box and Gray-Box Scenarios.}
Fig.~\ref{fig:effectiveness} shows that \MethodName{} consistently lowers post-fine-tuning accuracy compared with the Clean Fine-Tuning reference across the evaluated LVLMs and datasets, confirming that protected samples induce a reliable degradation signal during unauthorized training. The effect is most pronounced in the white-box setting on the surrogate Qwen3-VL-4B, where protected models often approach or fall below the Zero-Shot reference, largely eliminating the utility gain of clean fine-tuning. The degradation is especially clear on visually grounded benchmarks such as TextVQA, DocVQA, MMStar, and RealWorldQA, where performance depends heavily on image-text alignment. The protection also transfers to gray-box target models. Although the reduction varies with model similarity and task type, both \MethodName{}-BPH and \MethodName{}-CRS remain below the Clean Fine-Tuning reference in nearly all settings, indicating that the shortcut does not overfit to surrogate models. The effect is milder on ScienceQA and VQA-RAD, where language priors and domain knowledge can partially compensate for disrupted multimodal learning. The \textsc{Max} variants further strengthen protection in several cases, as they directly obstruct optimization by making protected samples hard to fit. Overall, \MethodName{} is effective under both white-box and more realistic gray-box scenarios.

\begin{figure}[t]
\centering
\includegraphics[width=\linewidth]{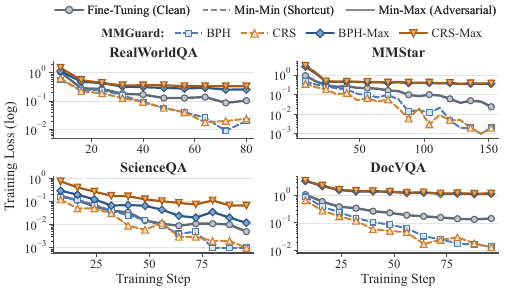}
\caption{Attacker training loss on protected data (log scale). Min-min variants converge to lower loss than \textsc{Clean FT}, indicating shortcut fitting, while \textsc{Max} variants plateau at higher loss, indicating training disruption.}
\label{fig:train}
\end{figure}

\begin{table*}[t]
\centering
\caption{Protection transferability in the no-box setting. Each cell reports test accuracy with the drop in parentheses.}
\label{tab:blackbox-transfer}
\scriptsize
\setlength{\tabcolsep}{0pt}
\newcommand{\tbp}{3pt}
\newcommand{\tbg}{9pt}
\newcommand{\tbx}{6pt}
\renewcommand{\arraystretch}{1.15}
\definecolor{tbdown}{HTML}{2C5689}
\definecolor{tbup}{HTML}{B4442F}
\definecolor{tbref}{HTML}{7A828E}
\definecolor{tbours}{HTML}{EFF4FB}
\definecolor{tbhead}{HTML}{1F2937}
\newcommand{\tdn}[1]{{\color{tbdown}\tiny\,($\downarrow$#1)}}
\newcommand{\tup}[1]{{\color{tbup}\tiny\,($\uparrow$#1)}}
\newcommand{\trf}[1]{{\color{tbref}\tiny\,($\downarrow$#1)}}
\newcommand{\tnc}{{\color{black!40}\tiny\,($\pm$0.0)}}
\newcommand{\tbest}[1]{\textbf{#1}}
\newcommand{\dshead}[1]{\multicolumn{10}{c}{\color{tbhead}\textbf{#1}}}
\newcommand{\dsheadg}[1]{\multicolumn{10}{c<{\hspace{\tbx}}}{\color{tbhead}\textbf{#1}}}
\newcommand{\mhead}[1]{\multicolumn{2}{>{\hspace{\tbp}}c<{\hspace{\tbp}}}{\color{tbhead}#1}}
\newcommand{\mheadl}[1]{\multicolumn{2}{>{\hspace{\tbp}}c<{\hspace{\tbg}}}{\color{tbhead}#1}}
\newcommand{\baseline}[1]{{\color{tbref}\textit{#1}}}

\resizebox{\textwidth}{!}{%
\begin{tabular}{l<{\hspace{\tbp}} *{4}{>{\hspace{\tbp}}rl<{\hspace{\tbp}}} >{\hspace{\tbp}}rl<{\hspace{\tbg}} *{4}{>{\hspace{\tbp}}rl<{\hspace{\tbp}}} >{\hspace{\tbp}}rl<{\hspace{\tbg}} *{5}{>{\hspace{\tbp}}rl<{\hspace{\tbp}}}}
\toprule
 & \dsheadg{RealWorldQA} & \dsheadg{ScienceQA} & \dshead{TextVQA} \\
\cmidrule(lr{12pt}){2-11}\cmidrule(lr{12pt}){12-21}\cmidrule(lr){22-31}
{\color{tbhead}\textbf{Method}} & \mhead{Llama} & \mhead{LLaVA} & \mhead{InternVL} & \mhead{Gemma} & \mheadl{GLM} & \mhead{Llama} & \mhead{LLaVA} & \mhead{InternVL} & \mhead{Gemma} & \mheadl{GLM} & \mhead{Llama} & \mhead{LLaVA} & \mhead{InternVL} & \mhead{Gemma} & \mhead{GLM} \\
\midrule
\baseline{Zero-Shot (Clean)} & 49.7 & \trf{6.9} & 53.6 & \trf{3.9} & 54.9 & \trf{5.9} & 36.6 & \trf{19.6} & 27.5 & \trf{39.8} & 84.5 & \trf{11.0} & 63.0 & \trf{16.5} & 97.5 & \trf{1.5} & 76.5 & \trf{8.0} & 88.5 & \trf{9.5} & 14.5 & \trf{67.5} & 51.5 & \trf{4.0} & 75.0 & \trf{1.5} & 54.5 & \trf{28.5} & 0.0 & \trf{85.0} \\
\baseline{Fine-Tuning (Clean)} & 56.6 &  & 57.5 &  & 60.8 &  & 56.2 &  & 67.3 &  & 95.5 &  & 79.5 &  & 99.0 &  & 84.5 &  & 98.0 &  & 82.0 &  & 55.5 &  & 76.5 &  & 83.0 &  & 85.0 &  \\
\midrule
Image                  & 61.4 & \tup{4.8} & 55.6 & \tdn{1.9} & 60.8 & \tnc & 51.6 & \tdn{4.6} & 62.7 & \tdn{4.6} & 93.5 & \tdn{2.0} & 77.5 & \tdn{2.0} & 98.5 & \tdn{0.5} & 82.0 & \tdn{2.5} & 99.0 & \tup{1.0} & 77.5 & \tdn{4.5} & 52.0 & \tdn{3.5} & 76.0 & \tdn{0.5} & 79.0 & \tdn{4.0} & 83.5 & \tdn{1.5} \\
Text                   & 57.5 & \tup{0.9} & 55.6 & \tdn{1.9} & 60.1 & \tdn{0.7} & 52.9 & \tdn{3.3} & \tbest{60.1} & \tdn{7.2} & 95.0 & \tdn{0.5} & 76.5 & \tdn{3.0} & 99.0 & \tnc & 80.5 & \tdn{4.0} & 98.5 & \tup{0.5} & 80.5 & \tdn{1.5} & 53.0 & \tdn{2.5} & 74.5 & \tdn{2.0} & 81.5 & \tdn{1.5} & 83.5 & \tdn{1.5} \\
Multimodal             & 56.9 & \tup{0.3} & 56.2 & \tdn{1.3} & 58.8 & \tdn{2.0} & 51.6 & \tdn{4.6} & 62.1 & \tdn{5.2} & 93.5 & \tdn{2.0} & 75.5 & \tdn{4.0} & 98.0 & \tdn{1.0} & 82.0 & \tdn{2.5} & 97.5 & \tdn{0.5} & 79.5 & \tdn{2.5} & 52.0 & \tdn{3.5} & 76.5 & \tnc & 80.0 & \tdn{3.0} & 82.5 & \tdn{2.5} \\
\addlinespace[2pt]
\rowcolor{tbours}\MethodName{}-BPH      & 52.9 & \tdn{3.7} & 55.6 & \tdn{1.9} & \tbest{54.9} & \tdn{5.9} & 50.3 & \tdn{5.9} & 60.8 & \tdn{6.5} & \tbest{90.0} & \tdn{5.5} & 73.5 & \tdn{6.0} & 97.5 & \tdn{1.5} & \tbest{78.5} & \tdn{6.0} & 95.0 & \tdn{3.0} & \tbest{76.0} & \tdn{6.0} & 51.0 & \tdn{4.5} & 74.0 & \tdn{2.5} & \tbest{75.5} & \tdn{7.5} & 82.0 & \tdn{3.0} \\
\rowcolor{tbours}\MethodName{}-CRS      & \tbest{51.7} & \tdn{4.9} & 54.4 & \tdn{3.1} & 56.2 & \tdn{4.6} & 51.0 & \tdn{5.2} & 61.4 & \tdn{5.9} & 92.0 & \tdn{3.5} & 74.5 & \tdn{5.0} & 97.5 & \tdn{1.5} & 79.5 & \tdn{5.0} & \tbest{93.5} & \tdn{4.5} & 78.0 & \tdn{4.0} & 50.5 & \tdn{5.0} & 74.5 & \tdn{2.0} & 77.0 & \tdn{6.0} & 82.0 & \tdn{3.0} \\
\rowcolor{tbours}\MethodName{}-BPH-Max  & 54.4 & \tdn{2.2} & 56.0 & \tdn{1.5} & 55.4 & \tdn{5.4} & \tbest{49.5} & \tdn{6.7} & 62.6 & \tdn{4.7} & 92.0 & \tdn{3.5} & \tbest{72.0} & \tdn{7.5} & \tbest{96.5} & \tdn{2.5} & 79.5 & \tdn{5.0} & 96.5 & \tdn{1.5} & 77.0 & \tdn{5.0} & \tbest{49.5} & \tdn{6.0} & 74.0 & \tdn{2.5} & 76.5 & \tdn{6.5} & 81.0 & \tdn{4.0} \\
\rowcolor{tbours}\MethodName{}-CRS-Max  & 53.6 & \tdn{3.0} & \tbest{53.4} & \tdn{4.1} & 57.3 & \tdn{3.5} & 51.4 & \tdn{4.8} & 63.2 & \tdn{4.1} & 91.5 & \tdn{4.0} & 76.0 & \tdn{3.5} & \tbest{96.5} & \tdn{2.5} & \tbest{78.5} & \tdn{6.0} & 95.0 & \tdn{3.0} & 78.5 & \tdn{3.5} & 50.0 & \tdn{5.5} & \tbest{73.0} & \tdn{3.5} & 78.5 & \tdn{4.5} & \tbest{80.5} & \tdn{4.5} \\
\midrule[\heavyrulewidth]
 & \dsheadg{MMStar} & \dsheadg{VQA-RAD} & \dshead{DocVQA} \\
\cmidrule(lr{12pt}){2-11}\cmidrule(lr{12pt}){12-21}\cmidrule(lr){22-31}
{\color{tbhead}\textbf{Method}} & \mhead{Llama} & \mhead{LLaVA} & \mhead{InternVL} & \mhead{Gemma} & \mheadl{GLM} & \mhead{Llama} & \mhead{LLaVA} & \mhead{InternVL} & \mhead{Gemma} & \mheadl{GLM} & \mhead{Llama} & \mhead{LLaVA} & \mhead{InternVL} & \mhead{Gemma} & \mhead{GLM} \\
\midrule
\baseline{Zero-Shot (Clean)} & 49.7 & \trf{11.0} & 34.3 & \trf{9.2} & 63.0 & \trf{5.0} & 43.3 & \trf{8.4} & 48.3 & \trf{22.7} & 0.0 & \trf{58.8} & 37.5 & \trf{14.6} & 53.9 & \trf{5.3} & 7.3 & \trf{44.7} & 0.0 & \trf{58.9} & 2.5 & \trf{66.5} & 11.0 & \trf{9.5} & 36.5 & \trf{8.5} & 51.0 & \trf{16.0} & 0.0 & \trf{76.5} \\
\baseline{Fine-Tuning (Clean)} & 60.7 &  & 43.5 &  & 68.0 &  & 51.7 &  & 71.0 &  & 58.8 &  & 52.1 &  & 59.2 &  & 52.0 &  & 58.9 &  & 69.0 &  & 20.5 &  & 45.0 &  & 67.0 &  & 76.5 &  \\
\midrule
Image                  & 60.3 & \tdn{0.4} & 43.3 & \tdn{0.2} & 67.7 & \tdn{0.3} & 48.7 & \tdn{3.0} & 67.7 & \tdn{3.3} & 55.0 & \tdn{3.8} & 49.5 & \tdn{2.6} & 57.4 & \tdn{1.8} & 50.6 & \tdn{1.4} & 52.3 & \tdn{6.6} & 65.5 & \tdn{3.5} & 16.0 & \tdn{4.5} & 41.5 & \tdn{3.5} & 63.0 & \tdn{4.0} & 71.0 & \tdn{5.5} \\
Text                   & 59.7 & \tdn{1.0} & 41.7 & \tdn{1.8} & 66.3 & \tdn{1.7} & 51.3 & \tdn{0.4} & 68.7 & \tdn{2.3} & 53.9 & \tdn{4.9} & 49.7 & \tdn{2.4} & 56.8 & \tdn{2.4} & 49.2 & \tdn{2.8} & 53.7 & \tdn{5.2} & 66.0 & \tdn{3.0} & 16.5 & \tdn{4.0} & 41.5 & \tdn{3.5} & 60.5 & \tdn{6.5} & 71.0 & \tdn{5.5} \\
Multimodal             & 58.0 & \tdn{2.7} & 41.0 & \tdn{2.5} & 67.7 & \tdn{0.3} & 50.3 & \tdn{1.4} & 67.7 & \tdn{3.3} & 55.4 & \tdn{3.4} & 48.8 & \tdn{3.3} & 55.9 & \tdn{3.3} & 50.1 & \tdn{1.9} & 52.3 & \tdn{6.6} & 66.0 & \tdn{3.0} & 16.0 & \tdn{4.5} & 40.5 & \tdn{4.5} & 62.0 & \tdn{5.0} & 73.5 & \tdn{3.0} \\
\addlinespace[2pt]
\rowcolor{tbours}\MethodName{}-BPH      & 56.3 & \tdn{4.4} & 39.0 & \tdn{4.5} & 65.0 & \tdn{3.0} & 45.7 & \tdn{6.0} & 65.7 & \tdn{5.3} & 52.1 & \tdn{6.7} & 46.2 & \tdn{5.9} & 55.4 & \tdn{3.8} & 47.5 & \tdn{4.5} & 52.1 & \tdn{6.8} & 63.5 & \tdn{5.5} & 13.5 & \tdn{7.0} & 41.0 & \tdn{4.0} & 60.5 & \tdn{6.5} & 68.5 & \tdn{8.0} \\
\rowcolor{tbours}\MethodName{}-CRS      & \tbest{54.0} & \tdn{6.7} & \tbest{38.7} & \tdn{4.8} & 65.3 & \tdn{2.7} & \tbest{45.3} & \tdn{6.4} & 66.3 & \tdn{4.7} & \tbest{51.2} & \tdn{7.6} & \tbest{45.8} & \tdn{6.3} & 56.1 & \tdn{3.1} & \tbest{47.0} & \tdn{5.0} & 52.1 & \tdn{6.8} & \tbest{63.0} & \tdn{6.0} & 14.5 & \tdn{6.0} & 40.5 & \tdn{4.5} & 61.0 & \tdn{6.0} & \tbest{68.0} & \tdn{8.5} \\
\rowcolor{tbours}\MethodName{}-BPH-Max  & 57.2 & \tdn{3.5} & 39.2 & \tdn{4.3} & \tbest{64.5} & \tdn{3.5} & 47.5 & \tdn{4.2} & 65.8 & \tdn{5.2} & 53.2 & \tdn{5.6} & 47.2 & \tdn{4.9} & 55.0 & \tdn{4.2} & 47.5 & \tdn{4.5} & 51.3 & \tdn{7.6} & 64.0 & \tdn{5.0} & \tbest{13.0} & \tdn{7.5} & 41.0 & \tdn{4.0} & \tbest{59.5} & \tdn{7.5} & \tbest{68.0} & \tdn{8.5} \\
\rowcolor{tbours}\MethodName{}-CRS-Max  & 56.2 & \tdn{4.5} & 39.2 & \tdn{4.3} & 65.2 & \tdn{2.8} & 48.8 & \tdn{2.9} & \tbest{65.5} & \tdn{5.5} & 52.3 & \tdn{6.5} & 46.7 & \tdn{5.4} & \tbest{54.6} & \tdn{4.6} & 49.9 & \tdn{2.1} & \tbest{50.8} & \tdn{8.1} & 64.5 & \tdn{4.5} & 14.0 & \tdn{6.5} & \tbest{40.0} & \tdn{5.0} & \tbest{59.5} & \tdn{7.5} & 70.5 & \tdn{6.0} \\
\bottomrule
\end{tabular}}
\end{table*}

\noindent\textbf{Training Dynamics on Protected Data.}
Fig.~\ref{fig:train} further explains the accuracy degradation by showing how protected data changes the attacker's optimization trajectory. The min-min variants, \MethodName{}-BPH and \MethodName{}-CRS, consistently drive the training loss below Clean Fine-Tuning, often by one to two orders of magnitude. This indicates that the attacker can fit the protected training set more easily than the clean one by exploiting a perturbation-induced shortcut. As a result, the low training loss does not transfer to clean-test accuracy. In contrast, the adversarial \textsc{Max} variants follow the opposite pattern: their losses remain substantially above Clean Fine-Tuning and plateau during training. This suggests a direct training-disruption effect, where the protected samples prevent effective empirical risk minimization rather than merely offering an easier shortcut. The attacker, therefore, cannot recover by simply training longer. Together, these two regimes reveal complementary protection mechanisms: min-min variants induce shortcut overfitting, while \textsc{Max} variants obstruct optimization. These dynamics support the design intuition in Section~\ref{sec:adversarial-objective} and explain the consistent effectiveness observed in Fig.~\ref{fig:effectiveness}. 

\subsection{Transferability Evaluation}
\label{sec:transferability}
In pratical, A defender  has no prior knowledge of the LVLM or fine-tuning recipe an attacker may employ. We evaluate transferability: (\emph{i}) across LVLMs in the no-box setting, and (\emph{ii}) across attacker fine-tuning strategies.

\noindent\textbf{Transferability Across LVLMs.}
We evaluate five no-box LVLMs with different architectures. Table~\ref{tab:blackbox-transfer} shows a clear gap between binding-aware and binding-agnostic protections. Single-modality Image and Text baselines, as well as the naive Multimodal combination, provide limited and unstable protection: their drops are often small and occasionally negative, meaning that protected fine-tuning can even improve clean-test accuracy over Clean Fine-Tuning. This suggests that input-space perturbation alone is insufficient for no-box transfer. When one modality remains clean, the attacker can still learn from the unprotected channel, causing the perturbation to behave more like regularization than a reliable protection signal. In contrast, all \MethodName{} variants produce positive accuracy drops across the evaluations. The gains are modest on saturated tasks such as ScienceQA, where Clean Fine-Tuning is near the ceiling, but are more pronounced on visually grounded tasks such as TextVQA, DocVQA, MMStar, and VQA-RAD. These results identify cross-modal binding as the key transferable component: \MethodName{} disrupts the attention route linking visual evidence, textual triggers, and answers, a structure broadly shared by LVLMs. Overall, \MethodName{} transfers consistently to unseen architectures and provides more reliable no-box protection than binding-agnostic baselines.

\begin{figure}[t]
\centering
\includegraphics[width=\linewidth]{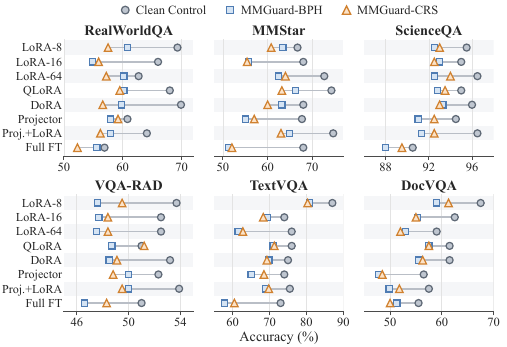}
\caption{Protection transferability of \MethodName{}-BPH and \MethodName{}-CRS across various fine-tuning strategies employed by attackers.}
\label{fig:finetuning}
\end{figure}

\noindent\textbf{Transferability Across Fine-Tuning Methods.}
A defense-aware attacker may change the fine-tuning recipe to bypass the protection signal. We fix the protected data and let the attacker use eight representative strategies: LoRA with ranks $r{=}8/16/64$, QLoRA, DoRA, projector-only tuning, projector\,+\, LLM LoRA, and full fine-tuning. Fig.~\ref{fig:finetuning} shows that both \MethodName{}-BPH and \MethodName{}-CRS consistently reduce accuracy compared with the corresponding Clean Fine-Tuning reference across datasets and recipes, indicating that the protection is not tied to a specific tuning method. The degradation is most evident on visually grounded datasets such as TextVQA and DocVQA. Low-capacity recipes, such as projector-only tuning, are particularly ineffective because their limited number of trainable parameters cannot bypass protection. Importantly, the effect persists under full fine-tuning, in which all model parameters are updated, indicating that increasing adaptation capacity alone cannot remove the protection signal. Overall, \MethodName{} transfers across diverse fine-tuning strategies. 

\subsection{Practicality Evaluation}
\label{sec:practicality}
Beyond effectiveness, a deployable defense must remain unobtrusive to legitimate users and feasible to apply, so we evaluate \MethodName{} on two practicality axes: (\emph{i}) multimodal stealthiness; and (\emph{ii}) algorithmic efficiency.

\noindent\textbf{Multimodal Stealthiness and Semantic Coherence.}
A practical data-side defense should preserve the utility of released data for human users. We evaluate image stealthiness using PSNR, SSIM, and LPIPS; text stealthiness using perplexity, edit distance, and BLEU; and semantic coherence using a 1-3 rubric covering image naturalness, text naturalness, image-text coherence, and human answerability. The rubric aggregates ratings from three human experts and three SOTA LLM judges (i.e., GPT-5.4, Gemini-3.1-Pro-Preview, and Claude-Sonnet-4-6). Table~\ref{tab:stealthiness-coherence} shows that the image-side distortion is minor: \MethodName{}-BPH and \MethodName{}-CRS achieve PSNR above $34$~dB, SSIM around $0.88$, and LPIPS $0.11$, matching the Image-only and naive Multimodal baselines. The main perceptual cost stems from text-trigger insertion, which increases perplexity and lowers text naturalness. Nevertheless, BLEU remains high at $0.93$, image-text coherence remains $2.14/3$, and human answerability remains $2.33/3$, indicating that the protected samples are still understandable and answerable to human users. Importantly, BPH and CRS introduce no additional observable stealthiness cost beyond the underlying image perturbation and text trigger, showing that the binding-disruption objective improves protection without further degrading usability.

\begin{table}[t]
\centering
\caption{Multimodal stealthiness and semantic coherence of protected data. Human-study metrics are evaluated by three human experts and three SOTA LLM judges on 3 scale.}
\label{tab:stealthiness-coherence}
\scriptsize
\setlength{\tabcolsep}{3.5pt}
\renewcommand{\arraystretch}{1.12}
\resizebox{\columnwidth}{!}{%
\begin{tabular}{@{}lrrrrrr@{}}
\toprule
\textbf{Metric} & \textbf{Clean} & \textbf{Image} & \textbf{Text} & \textbf{Multi.} & \textbf{BPH} & \textbf{CRS} \\
\midrule
\multicolumn{7}{@{}l}{\textit{Computational Metrics}} \\
\midrule
PSNR $\uparrow$ & $+\infty$ & 34.69 & $+\infty$ & 34.63 & 34.48 & 34.68 \\
SSIM $\uparrow$ & 1.00 & 0.88 & 1.00 & 0.88 & 0.87 & 0.88 \\
LPIPS $\downarrow$ & 0.00 & 0.11 & 0.00 & 0.11 & 0.11 & 0.11 \\
PPL $\downarrow$ & 39.12 & 39.12 & 89.72 & 89.24 & 91.91 & 88.28 \\
Edit Distance $\downarrow$ & 0.00 & 0.00 & 5.57 & 29.48 & 36.29 & 29.70 \\
BLEU $\uparrow$ & 1.00 & 1.00 & 0.98 & 0.93 & 0.93 & 0.93 \\
\midrule
\multicolumn{7}{@{}l}{\textit{Human Study / LLM-as-a-Judge}} \\
\midrule
Image Naturalness $\uparrow$ & 3.00 & 2.52 & 3.00 & 2.52 & 2.52 & 2.52 \\
Text Naturalness $\uparrow$ & 2.86 & 2.86 & 1.33 & 1.36 & 2.00 & 1.99 \\
Image-Text Coherence $\uparrow$ & 2.92 & 2.92 & 2.00 & 2.14 & 2.14 & 2.14 \\
Human Answerability $\uparrow$ & 2.73 & 2.60 & 1.99 & 2.33 & 2.33 & 2.33 \\
\bottomrule
\end{tabular}}
\end{table}

\noindent\textbf{Algorithmic Efficiency.}
The cost of \MethodName{} is dominated by surrogate LVLM forward and  backward passes, giving an overall time complexity of $O(RnM(Q+1)C)$ for $n$ samples, $M$ surrogates, $R$ outer rounds, and $Q$ inner steps, where $C$ is the cost of one forward/backward pass. The trigger search and the binding loss add only minor overhead, since the trigger and candidate set are small and the attention maps are reused from the forward pass. As a one-time, offline preprocessing step, \MethodName{} thus scales linearly with dataset size and adds no inference-time overhead. All experiments can be run on a single A100 GPU with 80 GB VRAM. Table~\ref{tab:runtime} reports the measured protection-generation time on 4 NVIDIA A100 GPUs with 80 GB VRAM. CRS costs about $1.3\times$ BPH due to its additional clean-reference forward pass, whereas each \textsc{Max} variant costs nearly the same as its min-min counterpart. Overall, \MethodName{} is efficient with a desirable one-time cost to deploy in practice.

\begin{table}[t]
\centering
\caption{Protection-generation time (min) of \MethodName{} variants on 4 NVIDIA A100 GPUs with 80 GB VRAM.}
\label{tab:runtime}
\scriptsize
\setlength{\tabcolsep}{4pt}
\renewcommand{\arraystretch}{1.05}
\begin{tabular}{@{}lrrrr@{}}
\toprule
\textbf{Dataset} & \textbf{BPH} & \textbf{CRS} & \textbf{BPH-Max} & \textbf{CRS-Max} \\
\midrule
RealWorldQA & 34.0 & 45.3 & 33.8 & 44.8 \\
MMStar      & 52.9 & 69.5 & 52.1 & 68.5 \\
VQA-RAD     & 90.2 & 117.8 & 92.2 & 119.3 \\
\bottomrule
\end{tabular}
\end{table}

\subsection{Adaptive Attack Evaluation}
\label{sec:adaptive}
An attacker aware of~\MethodName{} may seek to invalidate the protection. Following prior evaluations of perturbative protections~\cite{cao2023impress,honig2025adversarial}, we consider three adaptive strategies: (\emph{i})~data transformation and purification, which attempts to remove the protection signal from the inputs; (\emph{ii})~adversarial training, which attempts to make the model robust to it; and (\emph{iii})~data mixing, which dilutes it with clean data.

\begin{figure}[t]
\centering
\includegraphics[width=\linewidth]{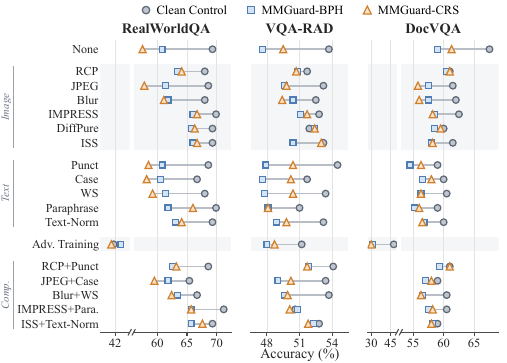}
\caption{Protection robustness of \MethodName{}-BPH and \MethodName{}-CRS under adaptive attacks, including data transformations, data purification, and adversarial training. Clean fine-tuning results are reported as a control.}
\label{fig:transformation}
\end{figure}

\noindent\textbf{Robustness Against Data Transformation.}
A defense-aware attacker may preprocess the scraped data to weaken the protection signal. We first evaluate standard transformations, including image-side operations (RCP, JPEG compression, and blurring), text-side normalizations (punctuation removal, case normalization, and whitespace normalization), and their compositions. Fig.~\ref{fig:transformation} reports the results on RealWorldQA, VQA-RAD, and DocVQA, where each protected model is compared with Clean Control under the same transformation. These transformations only modestly weaken the protection: \MethodName{} remains at or below the reference in all cases, with gaps of up to about $11$ points. Image-side operations, particularly RCP, are generally more effective than text-side normalizations, as they partially distort the perturbation, whereas the trigger consists of whole words that normalization leaves largely intact.

\begin{figure*}[th]
\centering
\includegraphics[width=\linewidth]{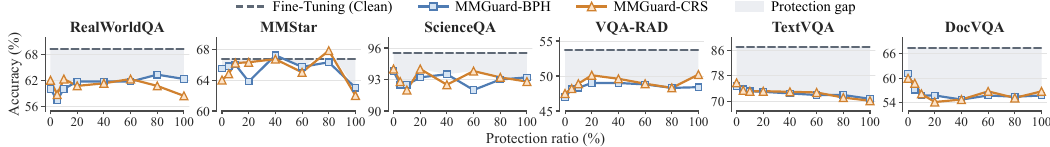}
\caption{Protection effectiveness under attacker-side data mixing. Lower curves indicate stronger protection.}
\label{fig:ratio}
\end{figure*}

\begin{figure*}[th]
\centering
\includegraphics[width=\linewidth]{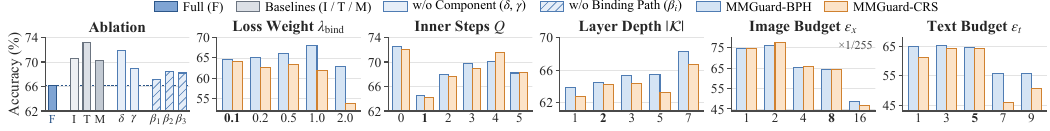}
\caption{Ablation study (leftmost) and parameter sensitivity of \MethodName{}, averaged over six datasets. Lower values indicate stronger protection.}
\label{fig:parameter}
\end{figure*}

\noindent\textbf{Robustness Against Data Purification.}
A stronger attacker may apply dedicated purification methods, including IMPRESS~\cite{cao2023impress}, DiffPure~\cite{nie2022diffusion}, and ISS~\cite{liu2023image} on the image side, paraphrasing and text normalization (Text-Norm) with Qwen3-8B on the text side, and their compositions. Purification is markedly stronger than standard transformations. Image-side purification narrows the gap to about $4$ points, with DiffPure showing the strongest effect on VQA-RAD. Text-side paraphrasing and normalization are less effective, leaving gaps of roughly $3$-$8$ points, as they only partially remove the trigger. Combining both sides, ISS+Text-Norm reduces the gap to below $4$ points across datasets. These results indicate that purifying both modalities can weaken but rarely fully eliminate the protection, while incurring non-negligible computational cost (e.g., running a diffusion model or an LLM over every training sample), showing that \MethodName{} is robust and promising in the real world.

\noindent\textbf{Robustness Against Adversarial Training.}
The attacker may instead fine-tune with adversarial training~\cite{madry2018towards} to become robust to the perturbation. However, this incurs a severe utility cost, reducing Clean Control accuracy by about $28$ points on RealWorldQA and $20$ points on DocVQA. Although adversarial training closes the protection gap on RealWorldQA, the resulting model is nearly $18$ points worse than simply fine-tuning on the protected data. On DocVQA, it even amplifies the protection, leaving \MethodName{}-BPH about $17$ points below Clean Control, and on VQA-RAD a gap of about $3$ points persists. Thus, adversarial training is not a viable bypass, as its utility cost outweighs its benefit.

\noindent\textbf{Robustness Against Data Mixing.}
A defense-aware attacker may attempt to dilute the protection signal by mixing protected samples with clean data collected from external sources. We evaluate this setting by varying the protection ratio from $0\%$ to $100\%$ and reporting clean-test accuracy after fine-tuning. Fig.~\ref{fig:ratio} summarizes the results. Across all six datasets, both \MethodName{}-BPH and \MethodName{}-CRS consistently remain below the Clean Fine-Tuning reference, showing that data mixing cannot fully eliminate the protection effect. The degradation is particularly clear on TextVQA and DocVQA, where increasing the protected portion further reduces clean-test accuracy, with full protection yielding the largest gap relative to clean fine-tuning. On RealWorldQA, MMStar, ScienceQA, and VQA-RAD, the curves are less strictly monotonic, but the protected models still stay below the clean reference across nearly all ratios. This indicates that even partial protected coverage can inject a usable pattern into the fine-tuning process. As the protection ratio increases, this shortcut receives stronger gradient support, but the exact accuracy trend depends on dataset difficulty and fine-tuning variance. Overall, \MethodName{} remains robust against attacker-side data mixing: clean data dilution may reduce the dosage of the protection signal, but does not restore clean fine-tuning performance. 

\subsection{Model Mechanism Analysis}
\label{sec:mechanism}
We analyze \emph{why} and \emph{how} \MethodName{} works through three perspectives: (\emph{i}) an ablation study isolating the contribution of each component, (\emph{ii}) a sensitivity analysis of the five core hyperparameters, and (\emph{iii}) a visualization-based analysis of the cross-modal binding mechanism.

\noindent\textbf{Ablation Study.}
The leftmost panel of Fig.~\ref{fig:parameter} shows that the full \MethodName{}-BPH achieves the strongest protection among all ablation variants. Using only image perturbations, only text triggers, or a naive multimodal combination yields noticeably higher accuracy, indicating that neither modality alone nor a simple combination is sufficient. This isolates the binding-disruption objective as the key component: it does not merely add image and text perturbations, but actively couples them into a shortcut that interferes with clean multimodal learning. In the full design, removing the image perturbation results in greater degradation than removing the text trigger, suggesting that the perturbation provides the dominant optimization signal, while the trigger acts as a discrete activator. Among the binding paths in Eq.~\eqref{eq:bph}, ablating the answer-anchored paths is more harmful. This is consistent with Theorem~\ref{thm:effectivenss-BPH}: the answer-anchored routes play a critical role in shifting the attention distribution.

\begin{figure}[th]
\centering
\includegraphics[width=\linewidth]{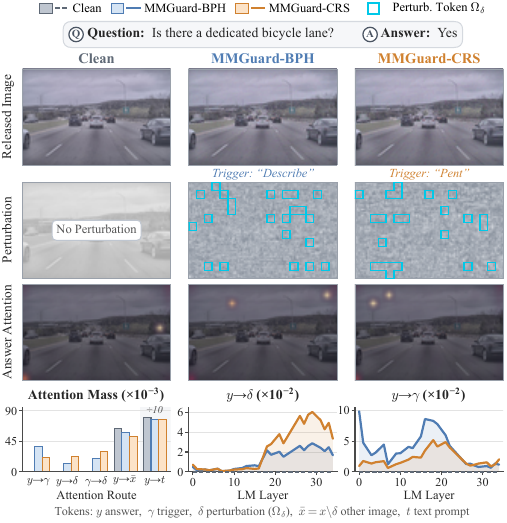}
\caption{Visualization of the cross-modal binding mechanism on a representative sample. Row~1: released images; Row~2: perturbation and perturbation tokens $\Omega_{\delta}$ (cyan boxes); Row~3: answer-token attention maps; Row~4: head-averaged attention mass and layerwise curves.}
\label{fig:mechanism}
\end{figure}

\noindent\textbf{Parameter Analysis.}
The remaining panels of Fig.~\ref{fig:parameter} examine the sensitivity of \MethodName{} to its main hyperparameters. Increasing the binding-loss weight $\lambda_{\mathrm{bind}}$ strengthens \MethodName{}-CRS, while \MethodName{}-BPH exhibits a non-monotonic trend, reflecting the difference between CRS's route-agnostic attention shifting and BPH's more prescriptive binding objective. The inner-loop step $Q$ is crucial: setting $Q{=}0$ removes the bilevel adaptation effect and substantially weakens protection, whereas a small positive value already captures most of the gain. For the gradient layer set $\mathcal{K}$, shallower layers provide stronger protection, suggesting that the relevant cross-modal binding behavior emerges early in multimodal fusion. Finally, both image and text budgets show clear strength--stealthiness trade-offs. Larger image perturbation budgets improve protection after a threshold, with $\epsilon_x{=}8/255$ serving as a practical default, while increasing the text trigger budget further strengthens protection but may reduce textual naturalness. Overall, these results support the design of \MethodName{}: effective protection requires both modalities, explicit cross-modal binding disruption, and a moderate perturbation budget.

\noindent\textbf{Cross-Modal Binding Mechanism Analysis.}
Fig.~\ref{fig:mechanism} visualizes the protection mechanism on a representative sample. The released image stays perceptually close to the clean input (Row~1), while the perturbation patches $\Omega_{\delta}$ (Row~2) form a structured area for the binding objective. Row~3 shows that attention is \emph{redirected} from the genuine visual evidence (e.g., bicycle lane) onto $\Omega_{\delta}$ in off-content regions (e.g., sky), and under \MethodName{}, answer-to-trigger, answer-to-perturbation, and trigger-to-perturbation attention all increase from the clean baseline, while attention to non-perturbed image tokens is suppressed. These patterns support the mechanism predicted by Theorem~\ref{thm:effectivenss-BPH}  because clean test inputs do not contain the learned trigger--perturbation binding, the model cannot reuse the shortcut acquired during protected fine-tuning. The layerwise curves also reveal different routing behaviors between the two variants. \MethodName{}-BPH concentrates answer-to-trigger attention in early-to-middle LM layers, consistent with the prescribed binding paths in Eq.~\eqref{eq:bph}; \MethodName{}-CRS induces stronger answer-to-perturbation attention in later layers, matching its route-agnostic objective in Eq.~\eqref{eq:crs}. Thus, both variants induce a shift between protected training and clean inference while via different attention routes.

\section{Limitations and Discussion}
\label{sec:discussion}

\noindent\textbf{Pre-Publication Scope.}
\MethodName{} protects data only when the defender controls it prior to release. Content that has already been published, scraped, or otherwise distributed without protection falls outside our threat model, and adversaries with access to such unprotected copies may circumvent the defense. Consequently, \MethodName{} should be viewed as one layer of broader defenses, complementing post-hoc defenses such as data-use auditing~\cite{huang2024general,sablayrolles2020radioactive}, machine unlearning~\cite{bourtoule2021machine}, and legal remedies.

\noindent\textbf{Pretraining-from-Scratch Threats.}
We focus on the dominant practical setting in which adversaries fine-tune public LVLM checkpoints on scraped data. We do not evaluate adversaries that pretrain LVLMs from scratch using corpora dominated by protected samples. Since this is economically impractical for the long-tail data considered in this paper, organizations with substantial computational resources and access to a large corpus may be able to extend this setting.

\noindent\textbf{Computational Overhead.}
Generating unlearnable examples requires solving the multi-objective optimization in Eq.~\eqref{eq:joint-constrained-objective}, jointly optimizing image perturbations, discrete triggers, and surrogate adaptation. This one-time pre-release cost scales linearly with dataset size. Improving scalability through techniques such as universal perturbation generators, shared surrogate caches, and other optimization accelerations remains a future direction in web-scale deployment.

\noindent\textbf{Stealthiness of Text Triggers.}
The primary stealthiness cost arises in the text modality. Although BLEU and human answerability remain largely preserved (Table~\ref{tab:stealthiness-coherence}), inserted triggers inevitably increase perplexity and reduce naturalness. This stems from optimizing trigger effectiveness under discrete admissibility constraints. We leave fluency-aware objectives and vocabulary filters to future work.

\noindent\textbf{Strong Purification.}
Although \MethodName{} withstands standard transformations and adversarial training, dedicated purification of both modalities (e.g., DiffPure or ISS combined with text normalization) can narrow the protection gap (Section~\ref{sec:adaptive}), albeit at a non-negligible computational cost. Designing perturbations that remain more effective under such purification, e.g., by optimizing through purification pipelines, is an important direction for future work.

\section{Conclusion}
We presented \MethodName{}, the first proactive framework for protecting multimodal data against unauthorized fine-tuning of LVLMs. By jointly optimizing bounded image perturbations and discrete text triggers with a cross-modal binding-disruption objective, \MethodName{} induces protection-specific shortcuts that do not transfer to clean evaluation. We formalized this problem within a practical threat model, established a theoretical analysis of the trigger search and binding mechanism, and instantiated the framework as a constrained multi-objective optimization problem over surrogate LVLMs. Experiments across six datasets and nine open-source LVLMs demonstrate that \MethodName{} is effective across white-box, gray-box, and no-box settings, while preserving perceptual fidelity and remaining robust to various fine-tuning recipes and adaptive attacks. These results position \MethodName{} as a structural primitive for proactive multimodal data ownership, contributing to more responsible governance of public multimodal content.






\bibliographystyle{IEEEtran}
\bibliography{reference}
%



\appendices

\section{Optimality of Trigger Updates}
\label{app:hotflip-screening-proof}

\begin{lemma}[Trigger Updates Optimality]
\label{lem:hotflip-screening}
Suppose $\ell_{\mathrm{prot}}$ is $L$-smooth as a function of the embedding at trigger position $j$, and let $R=\max_{v\in\mathcal{V}_{\mathrm{adm}}}\|e_v-e_{\gamma_{i,j}}\|$. Let $v^{\dagger}\in\arg\min_{v\in\mathcal{V}_{\mathrm{adm}}}\ell_{\mathrm{prot}}\!\bigl(\gamma_i^{[j=v]}\bigr)$ be the exact best admissible substitution, and let $v_{i,j}^{\star}$ be the candidate selected by Eq.~\eqref{eq:hotflip-verify} from the top-$c$ screen $\mathcal{V}^{\mathrm{cand}}_{i,j}$ ranked by $s_{i,j}(\cdot)$. Then
\begin{equation}
    \ell_{\mathrm{prot}}\!\bigl(\gamma_i^{[j=v_{i,j}^{\star}]}\bigr)
    -\ell_{\mathrm{prot}}\!\bigl(\gamma_i^{[j=v^{\dagger}]}\bigr)
    \;\le\; LR^{2}.
\label{eq:hotflip-optimality}
\end{equation}
\end{lemma}

\noindent\emph{Proof.}
For position $j$ of the trigger $\gamma_i$, let $e=e_{\gamma_{i,j}}$ denote its current embedding and write $\ell(e)=\ell_{\mathrm{prot}}(\gamma_i)$ as a function of that embedding with all other variables fixed. Define the candidate-induced loss change $\Delta_{i,j}(v)=\ell_{\mathrm{prot}}(\gamma_i^{[j=v]})-\ell_{\mathrm{prot}}(\gamma_i)$, and recall the linear score $s_{i,j}(v)=(e_v-e)^{\top}\nabla_{e}\ell$ from Eq.~\eqref{eq:hotflip-score}. By the descent-lemma form of $L$-smoothness along the segment from $e$ to $e_v$, for every $v\in\mathcal{V}_{\mathrm{adm}}$,
\begin{equation}
    \bigl|\Delta_{i,j}(v)-s_{i,j}(v)\bigr|
    \;\le\;
    \tfrac{L}{2}\,\|e_v-e\|^{2}
    \;\le\;
    \tfrac{LR^{2}}{2}.
\label{eq:appendix-smoothness}
\end{equation}
\emph{Case 1: $v^{\dagger}\in\mathcal{V}^{\mathrm{cand}}_{i,j}$.}
The verification step in Eq.~\eqref{eq:hotflip-verify} exactly minimizes $\ell_{\mathrm{prot}}$ over $\mathcal{V}^{\mathrm{cand}}_{i,j}$, so Eq.~\eqref{eq:hotflip-optimality} holds with gap zero.
\emph{Case 2: $v^{\dagger}\notin\mathcal{V}^{\mathrm{cand}}_{i,j}$.}
By the top-$c$ rule, every $w\in\mathcal{V}^{\mathrm{cand}}_{i,j}$ has $s_{i,j}(w)\le s_{i,j}(v^{\dagger})$, so by Eq.~\eqref{eq:appendix-smoothness},
\begin{align}
    \Delta_{i,j}(w)
    &\le s_{i,j}(w)+\tfrac{LR^{2}}{2}
    \le s_{i,j}(v^{\dagger})+\tfrac{LR^{2}}{2}
    \nonumber\\
    &\le \Delta_{i,j}(v^{\dagger})+LR^{2}.
\end{align}
Since $v_{i,j}^{\star}$ exactly minimizes $\ell_{\mathrm{prot}}$ on the shortlist, $\Delta_{i,j}(v_{i,j}^{\star})\le\Delta_{i,j}(w)\le\Delta_{i,j}(v^{\dagger})+LR^{2}$. Removing the common offset $\ell_{\mathrm{prot}}(\gamma_i)$ recovers Eq.~\eqref{eq:hotflip-optimality}. \hfill$\square$

\section{Contrastive Form of the Attention-Mass Loss}
\label{app:routing-contrastive-proof}

\begin{proposition}[Contrastive Learning for Attention Mass Alignment]
\label{prop:routing-contrastive}
For each $k\in\mathcal{K}$, let $\mathcal{B}_{k}=\{b\in\Omega:\bar{A}_{k,b}(\mathcal{S})>0\}$. If $\mathcal{R}\subseteq\mathcal{B}_{k}$ for all $k\in\mathcal{K}$, write each $\bar{A}_{k}(\mathcal{S})$ on its support as a softmax over effective logits $z_k$. Then the attention-mass loss in Eq.~\eqref{eq:routing-cost} satisfies
\begin{align}
    \ell_{\mathrm{mass}}(\mathcal{S},\mathcal{R})
    &=
    \frac{1}{|\mathcal{K}|}
    \sum_{k\in\mathcal{K}}
    \ell_{\mathrm{NCE}}^{(k)}(\mathcal{S},\mathcal{R})
    - \log|\mathcal{R}|,
    \label{eq:routing-contrastive}\\
    \ell_{\mathrm{NCE}}^{(k)}(\mathcal{S},\mathcal{R})
    &=
    \frac{1}{|\mathcal{R}|}
    \sum_{r\in\mathcal{R}}
    -\log
    \frac{\exp(z_{k,r})}
    {\sum_{b\in\mathcal{B}_{k}}\exp(z_{k,b})}.
    \label{eq:routing-nce}
\end{align}
\end{proposition}

\noindent\emph{Proof.}
Fix any $k\in\mathcal{K}$. If any $r\in\mathcal{R}$ has $\bar{A}_{k,r}(\mathcal{S})=0$, then $\mathrm{D}_{\mathrm{KL}}(U_{\mathcal{R}}\|\bar{A}_{k}(\mathcal{S}))=+\infty$, so the finite case requires $\mathcal{R}\subseteq\mathcal{B}_{k}$ for every layer. Since $U_{\mathcal{R}}$ is uniform on $\mathcal{R}$,
\begin{align}
    \mathrm{D}_{\mathrm{KL}}\!\bigl(U_{\mathcal{R}}\|\bar{A}_{k}(\mathcal{S})\bigr)
    &=
    \sum_{b\in\Omega}
    U_{\mathcal{R}}(b)
    \log
    \frac{U_{\mathcal{R}}(b)}{\bar{A}_{k,b}(\mathcal{S})}
    \nonumber\\
    &=
    \frac{1}{|\mathcal{R}|}
    \sum_{r\in\mathcal{R}}
    -\log \bar{A}_{k,r}(\mathcal{S})
    -
    \log|\mathcal{R}|.
\end{align}
Substituting the softmax form
$\bar{A}_{k,r}(\mathcal{S})=\exp(z_{k,r})/\sum_{b\in\mathcal{B}_{k}}\exp(z_{k,b})$
into the last line gives Eq.~\eqref{eq:routing-nce} for layer $k$. Averaging the resulting identity over $k\in\mathcal{K}$ gives Eq.~\eqref{eq:routing-contrastive}. Therefore, minimizing the finite attention-mass loss is equivalent, up to the constant $-\log|\mathcal{R}|$, to minimizing an averaged multi-positive contrastive objective whose negatives are the non-target tokens in each supported LVLM context. \hfill$\square$

\section{Proof of Bridge Path Hijack Effectiveness}
\label{app:shortcut-gap-proof}

\noindent\emph{Proof of Theorem~\ref{thm:effectivenss-BPH}.}
Fix sample $i$, a protection-induced target set $\mathcal{R}_{i}\in\{\Omega_{\gamma,i},\,\Omega_{\delta,i}\}$, and any layer $k\in\mathcal{K}$. Write $P_{k}=\bar{A}^{\mathrm{p}}_{k,i}(\Omega_y)$ and $Q_{k}=\bar{A}^{\mathrm{c}}_{k,i}(\Omega_y)$, both viewed as probability distributions on the common, zero-padded token universe in which the positions in $\mathcal{R}_{i}$ receive mass $0$ on the clean side. This follows from the input-dependent definition of $\mathcal{R}_{i}$: trigger tokens are absent at clean evaluation, and the perturbation-magnitude ranking that selects $\Omega_{\delta,i}$ in Eq.~\eqref{eq:perturbation-token-set} is degenerate when $\delta_{i}=0$, so no token is selected.

\emph{Step 1: Pinsker on the protected attention.}
By Eq.~\eqref{eq:routing-cost}, $\frac{1}{|\mathcal{K}|}\sum_{k\in\mathcal{K}}\mathrm{D}_{\mathrm{KL}}(U_{\mathcal{R}_{i}}\|P_{k})=\ell_{\mathrm{mass}}(\Omega_{y,i},\mathcal{R}_{i})\le\eta$. Pinsker's inequality per layer,
$\mathrm{TV}(U_{\mathcal{R}_{i}},P_{k})\le\sqrt{\mathrm{D}_{\mathrm{KL}}(U_{\mathcal{R}_{i}}\|P_{k})/2}$, and Jensen's inequality for the concave $u\mapsto\sqrt{u}$ yield
\begin{equation}
    \frac{1}{|\mathcal{K}|}\sum_{k\in\mathcal{K}}
    \mathrm{TV}\!\bigl(U_{\mathcal{R}_{i}},P_{k}\bigr)
    \;\le\;
    \sqrt{\eta/2}.
\label{eq:appendix-pinsker}
\end{equation}

\emph{Step 2: Target-mass concentration.}
For distributions $P,U$ on a finite set, $\mathrm{TV}(P,U)\ge|P(B)-U(B)|$ for every event $B$. Choosing $B=\mathcal{R}_{i}$ with $U_{\mathcal{R}_{i}}(\mathcal{R}_{i})=1$ gives $\mathrm{TV}(U_{\mathcal{R}_{i}},P_{k})\ge 1-P_{k}(\mathcal{R}_{i})$, so Eq.~\eqref{eq:appendix-pinsker} yields
\begin{equation}
    \frac{1}{|\mathcal{K}|}\sum_{k\in\mathcal{K}}
    P_{k}(\mathcal{R}_{i})
    \;\ge\;
    1-\sqrt{\eta/2}.
\label{eq:appendix-target-mass}
\end{equation}
On average over layers, at least a $1-\sqrt{\eta/2}$ fraction of the protected answer-token attention mass thus lies on $\mathcal{R}_{i}$.

\emph{Step 3: Forced redistribution at clean evaluation.}
Because $\mathcal{R}_{i}$ is empty on clean inputs, $Q_{k}(\mathcal{R}_{i})=0$ for every $k$ on the zero-padded clean side. Applying the same event-based bound to $(P_{k},Q_{k})$ with $B=\mathcal{R}_{i}$ gives $\mathrm{TV}(P_{k},Q_{k})\ge|P_{k}(\mathcal{R}_{i})-Q_{k}(\mathcal{R}_{i})|=P_{k}(\mathcal{R}_{i})$. Averaging over $k\in\mathcal{K}$ and using Eq.~\eqref{eq:appendix-target-mass} gives Eq.~\eqref{eq:tv-shift-bound}. The argument is uniform in the choice of $\mathcal{R}_{i}\in\{\Omega_{\gamma,i},\,\Omega_{\delta,i}\}$, so the bound holds for both the $\beta_{2}$ and $\beta_{3}$ terms of Eq.~\eqref{eq:bph}. \hfill$\square$

\begin{table*}[t]
\centering
\caption{Human evaluation rubric for protected multimodal examples.}
\label{tab:human-evaluation-rubric}
\scriptsize
\renewcommand{\arraystretch}{1.05}
\begin{tabular}{@{}>{\raggedright\arraybackslash}p{0.11\textwidth}>{\raggedright\arraybackslash}p{0.86\textwidth}@{}}
\toprule
\textbf{Dimension} & \textbf{Criteria} \\
\midrule
\textbf{Image Naturalness} &
\textbf{1:} The image contains clearly visible artifacts, corruption, or unnatural patterns that interfere with normal perception.\par
\textbf{2:} The image is generally recognizable and usable, but contains mild visual artifacts, slight distortion, or subtle abnormal patterns.\par
\textbf{3:} The image appears visually natural, clean, and indistinguishable from a normal, unmodified image. \\
\midrule
\textbf{Text Naturalness} &
\textbf{1:} The text is grammatically flawed, difficult to read, or contains clearly random, code-like, gibberish, or strongly suspicious tokens that noticeably disrupt natural reading.\par
\textbf{2:} The text remains understandable but includes minor awkwardness, unusual wording, single-word substitutions, or inserted recognizable words/proper nouns/phrase-like fragments that are stylistically abnormal but do not prevent comprehension.\par
\textbf{3:} The text is fluent, coherent, and appears naturally written without noticeable suspicious, irrelevant, or stylistically abnormal content. \\
\midrule
\textbf{Image-Text Coherence} &
\textbf{1:} The text is largely inconsistent with the image, refers to absent or contradictory visual content, or forms an incoherent multimodal pair.\par
\textbf{2:} The text is partially related to the image, but some visual references are vague, incomplete, or only weakly supported by the image.\par
\textbf{3:} The text is clearly and semantically consistent with the image, and the image provides appropriate visual evidence for the text or question. \\
\midrule
\textbf{Human Answerability} &
\textbf{1:} The question or task cannot be answered reliably from the protected pair due to visual degradation, textual ambiguity, or image-text mismatch.\par
\textbf{2:} The question or task is answerable, but with noticeable uncertainty caused by mild ambiguity, reduced clarity, or incomplete visual/textual evidence.\par
\textbf{3:} The question or task is clearly answerable from the protected pair, with sufficient visual and textual information for a confident human response. \\
\bottomrule
\end{tabular}
\end{table*}

\section{Dataset Details}
\label{app:dataset-details}

\textbf{RealWorldQA}~\cite{realworldqa2024} asks free-form questions about real-world visual attributes such as object counts, relative position, traffic lights, road layout, and spatial relations.
\textbf{MMStar}~\cite{chen2024we} contains vision-indispensable multiple-choice samples, filtered and manually reviewed from 22,401 candidates, covering six core capabilities and 18 detailed axes.
\textbf{ScienceQA}~\cite{lu2022learn} contains multiple-choice questions from elementary and high-school science curricula, spanning 26 topics, 127 categories, and 379 skills.
\textbf{VQA-RAD}~\cite{lau2018dataset} contains radiology image-question-answer triplets over 314 images, with both open-ended and close-ended yes/no questions.
\textbf{TextVQA}~\cite{singh2019towards} contains questions on 28,408 natural images whose answers require reading and reasoning about scene text.
\textbf{DocVQA}~\cite{mathew2021docvqa} contains free-form questions over more than 12,000 document images, where document structure and layout are often important for answering.

We split each dataset into training and test sets with an 80/20 ratio. RealWorldQA (612/153), MMStar (1,200/300), and \mbox{VQA-RAD} (1,793/451) use all examples, while for the large-scale ScienceQA, TextVQA, and DocVQA, we sample 1,000 examples each (800/200 train/test).

\section{Training Details}
\label{app:training-details}

\noindent\textbf{Surrogate Fine-Tuning (Inner Loop).}
Table~\ref{tab:training-details} lists all hyperparameters. Within each outer round of Algorithm~\ref{alg:mmguard-optimization}, each surrogate is adapted to the protected samples by $Q$ LoRA steps (rank~$8$ for Qwen3-VL-4B-Instruct and $16$ for MiniCPM-V-4) on all attention and MLP projections. Adapters are discarded after each round and never released.

\noindent\textbf{Protection Optimization (Outer Loop).}
The image perturbation $\delta_i$ is updated by PGD and projected onto $\mathcal{C}^{\delta}_{i}$. The text trigger $\gamma_i$ is updated by ranking candidate tokens with the HotFlip score in Eq.~\eqref{eq:hotflip-score}, retaining the \mbox{top-$K$} candidates per slot, and selecting the one that minimizes the joint loss as in Eq.~\eqref{eq:hotflip-verify}. The admissible vocabulary $\mathcal{V}_{\mathrm{adm}}$ is restricted to ASCII alphabetic words with a leading whitespace, excluding numbers and punctuation. \mbox{\MethodName{}-BPH} uses unit path weights $\beta_1{=}\beta_2{=}\beta_3{=}1$ in Eq.~\eqref{eq:bph}, and \mbox{\MethodName{}-CRS} replaces the three paths with Eq.~\eqref{eq:crs}. The \textsc{Max} variants optimize Eq.~\eqref{eq:adversarial-loss} instead of Eq.~\eqref{eq:binding-loss}. The no-box ensemble uses \mbox{Qwen3-VL-4B-Instruct} and \mbox{MiniCPM-V-4} with uniform weights $\omega_m{=}1/M$.

\noindent\textbf{Attacker Fine-Tuning and Evaluation.}
The attacker fine-tunes each target LVLM on the protected training split, and we report accuracy on the clean test split; adaptive attacks follow Section~\ref{sec:adaptive}. For multiple-choice datasets (MMStar, ScienceQA), a prediction is correct if it matches the ground-truth option letter after normalization (case folding, whitespace stripping, and prefix removal). Other datasets use normalized exact match and accept any listed answer aliases.

\begin{table}[t]
\centering
\caption{Hyperparameters of \MethodName{} and the attacker.}
\label{tab:training-details}
\footnotesize
\renewcommand{\arraystretch}{1.05}
\setlength{\tabcolsep}{4pt}
\begin{tabular}{@{}>{\raggedright\arraybackslash}p{0.38\linewidth}>{\raggedright\arraybackslash}p{0.58\linewidth}@{}}
\toprule
\textbf{Group / Parameter} & \textbf{Values} \\
\midrule
\multicolumn{2}{@{}l}{\emph{Surrogate fine-tuning (inner loop)}} \\
\midrule
LoRA rank / dropout / modules & $8$ or $16$ / $0.1$ / \texttt{all} \\
Learning rate / batch size & $\{1{\times}10^{-5}, 1{\times}10^{-4}\}$ / $2$ \\
Inner steps $Q$ & $\{0, 1, 2, 3, 4, 5\}$ \\
\midrule
\multicolumn{2}{@{}l}{\emph{Protection optimization (outer loop)}} \\
\midrule
Outer rounds $R$ & $\{1, 3, 5, 7\}$ \\
Image budget $\epsilon_x$ ($\times 1/255$) & $\{1, 2, 4, 8, 16\}$ \\
PGD step $\alpha$ / iterations & $1/255$ / $\{0, 1, 2, 5, 8\}$ \\
Trigger length $\epsilon_t$ (tokens) & $\{1, 3, 5, 7, 9\}$ \\
HotFlip top-$K$ candidates & $\{8, 16, 64, 128, 512\}$ \\
$\lambda_{\mathrm{train}}$ / $\lambda_{\mathrm{bind}}$ & $1.0$ / $\{0.1, 0.2, 0.5, 1.0, 2.0\}$ \\
Binding layers (first $|\mathcal{K}|$) &$\{1, 2, 3, 4, 5, 6, 7\}$ \\
\midrule
\multicolumn{2}{@{}l}{\emph{Attacker SFT}} \\
\midrule
Fine-tuning recipe & LoRA, QLoRA, DoRA, Projector, Projector+LLM LoRA, Full FT \\
LoRA rank & $\{8, 16, 64\}$ \\
Learning rate & $\{1, 4, 10, 30\}{\times}10^{-5}$ \\
Epochs / batch size & $\{1, 2, 4, 8, 10\}$ / $\{1, 2, 4\}$ \\
Scheduler / warmup ratio & cosine / $0.1$ \\
\bottomrule
\end{tabular}
\end{table}

\begin{figure*}[th]
\centering
\includegraphics[width=\linewidth]{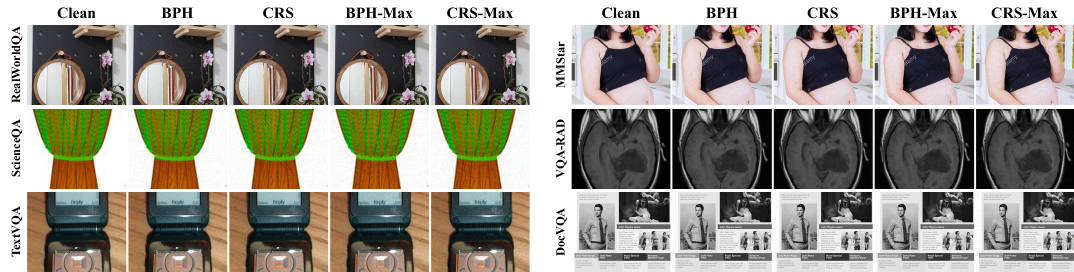}
\caption{Protection examples on six benchmarks: one random sample per block under the clean setting and the four \MethodName{} variants.}
\label{fig:protection-examples}
\end{figure*}

\section{Human Evaluation Rubric}
\label{app:human-evaluation}
We rate protected examples along the four dimensions in Table~\ref{tab:human-evaluation-rubric}, each on a three-point ordinal scale. The study covers $60$ protected image-text pairs ($10$ random samples per dataset), each independently rated by three human experts and three commercial LLM judges (GPT-5.4, \mbox{Gemini-3.1-Pro-Preview}, and \mbox{Claude-Sonnet-4-6}); Table~\ref{tab:stealthiness-coherence} reports the mean per-dimension score across the six raters. The judge prompt gives the image, the text, and each dimension's definition and criteria, asks for an integer score per dimension (the lower level if ambiguous), and requires JSON-only output with scores and a brief reason.

\section{Protection Examples}
\label{app:protection-examples}
Fig.~\ref{fig:protection-examples} shows protected samples of the \mbox{\MethodName{}} at the default budget $\epsilon_x{=}8/255$. Residual differences appear as low-amplitude texture, while key content for human understanding is preserved, e.g., lane geometry in RealWorldQA, radiographic contrast in VQA-RAD, and layout in DocVQA. Thus, \MethodName{} remains stealthy across domains.
\section{Ethics Considerations}
\label{app:ethics}
\MethodName{} is a defensive tool that allows data owners to reduce the fine-tuning value of image-text content \emph{they own} before publishing it online.

\noindent\textbf{Data and Human Subjects.}
All datasets are public benchmarks and all LVLMs are open-weight checkpoints, used under respective licenses. The three raters in our stealthiness study were members of the research team who annotated only the protected samples, and we collected no information about the raters themselves. Since the raters acted as investigators, the study does not constitute human subjects research and did not require IRB review.

\noindent\textbf{Live Systems.}
Our only interaction with external systems was about 180 API queries to three commercial LLM judges, made under the providers' terms of service and containing only protected public-benchmark samples. We released no protected data and did not interfere with deployed systems. All fine-tuning ran on our own hardware.

\noindent\textbf{Misuse and Mitigations.}
We identify three potential risks. First, an adversary with write access to a shared dataset could inject protected samples to degrade models trained on it (Section~\ref{sec:adaptive}). Curators can mitigate this via provenance tracking, data filtering, and validation on clean held-out data. Second, benign developers might unknowingly train on protected data. Because \MethodName{} is an untargeted availability defense that implants no input-specific behavior, standard pre-deployment validation reveals its effect, and data owners can further signal their preferences through licenses or opt-out mechanisms. Third, the binding objective could be repurposed for targeted backdoors. However, \MethodName{} leaves target responses unchanged, we neither develop nor evaluate such variants, and multimodal backdoor attacks are already well established~\cite{carlini2021extracting}. Finally, our adaptive attacks use only standard, public techniques.

\noindent\textbf{Disclosure and Release.}
We target a general property of fine-tuning on unvetted data that prior work has documented~\cite{huang2021unlearnable,fowl2021adversarial}, so coordinated disclosure does not apply. \MethodName{} gives data owners proactive control over their content, and its misuse risk is comparable to that of existing unlearnable-example methods, so we believe the benefits of publication outweigh the risks. We will release the code with documentation of intended use, limitations, and detection cues. We will redistribute protected samples only where the original dataset licenses permit such.

\section{Use of Generative AI}
Generative AI was used for editorial purposes in this manuscript, and all outputs were inspected by the authors to ensure accuracy and originality. We also used Opus-5.5 and GPT-6.1-Sol as coding assistants. The authors reviewed and tested all AI-assisted code, and all research ideas, methods, and analyses were developed by the authors. Generative AI is also part of our methodology. Open-weight LVLMs are the object of study, serving as surrogates and attackers. In addition, three commercial LLMs (GPT-5.4, \mbox{Gemini-3.1-Pro-Preview}, and \mbox{Claude-Sonnet-4-6}) serve as judges in the stealthiness evaluation (Appendix~\ref{app:human-evaluation}). To mitigate bias and ensure reproducibility, we gave them the same fixed rubric as the human raters, required structured outputs, reported exact model versions and the prompt, and averaged their scores with the human ratings. Our experiments used only public benchmarks, involved no new data collection, and were designed to limit compute. We used 2B-11B models and LoRA fine-tuning by default. Protection is generated once, offline, on multiple A100 GPUs with 80GB VRAM each, and the LLM judges received only about 180 queries.

\end{document}